\documentclass[11pt]{article}

\usepackage{graphicx} % Required for inserting images

\usepackage{amsmath, amsfonts, amssymb, amsthm}

\usepackage{CJKutf8}
\usepackage{fullpage}

\usepackage{graphicx}
\graphicspath{{figs/}}
\usepackage{amsmath,amssymb,amsfonts, amsthm}
\usepackage{xcolor}
\usepackage{xspace}
\usepackage{thm-restate}
\usepackage{mathtools}
\usepackage{algorithmicx} 

\usepackage[colorlinks,linkcolor=blue,filecolor=blue,citecolor=blue,urlcolor=blue,pagebackref]{hyperref}

\usepackage{soul}

\usepackage{algorithm}
\usepackage{algpseudocode}
\usepackage{bbold}

\algtext*{EndWhile}
\algtext*{EndIf}
\algtext*{EndFor}
\algrenewcommand\algorithmicrequire{\textbf{Input:}}
\algrenewcommand\algorithmicensure{\textbf{Output:}}

\usepackage[colorinlistoftodos,textsize=tiny,textwidth=2cm,color=red!25!white,obeyFinal]{todonotes}

\usepackage{enumitem}
\setlist{leftmargin=*, topsep=3pt, itemsep=3pt}

\newcommand{\R}{\mathbb{R}}
\newcommand{\cost}{\mathrm{cost}}

\newtheorem{theorem}{Theorem}[section]
\newtheorem{lemma}[theorem]{Lemma}
\newtheorem{definition}[theorem]{Definition}
\newtheorem{claim}[theorem]{Claim}
\newtheorem{coro}[theorem]{Corollary}

\DeclareMathOperator*{\E}{{\mathbb{E}}}

\newcommand{\calC}{{\mathcal{C}}}

\newcommand{\poly}{{\mathrm{poly}}}

\newcommand{\opt}{{\mathrm{opt}}}

\newcommand{\ball}{\mathrm{Ball}}
\newcommand{\profit}{\mathrm{profit}}
\newcommand{\onefive}{\frac{1}{5}}
\newcommand{\twofive}{\frac{2}{5}}
\newcommand{\threefive}{\frac{3}{5}}
\newcommand{\fourfive}{\frac{4}{5}}
\newcommand{\topk}{\mathrm{top}_k}
\newcommand{\order}{\mathrm{order}}

\allowdisplaybreaks

\newcommand{\OPT}{\mathrm{OPT}}

\begin{document}\begin{CJK*}{UTF8}{gbsn}

\title{Simultaneously Approximating All Norms for Massively Parallel Correlation Clustering}

\author{ 
{Nairen Cao \footnote{The work of NC was supported by NSF grant CCF-2008422.}} \\ Department of Computer Science \\ Boston College \\ Chestnut Hill, MA, United States \\ 
\href{mailto:nc1827@nyu.edu}{nc1827@nyu.edu} 
\and Shi Li \footnote{The work of SL and JY was supported by the State Key Laboratory for Novel Software Technology, and the New Cornerstone Science Laboratory.} \\ School of Computer Science, \\ Nanjing University, \\ Nanjing, Jiangsu Province, China \\ \href{mailto:shili@nju.edu.cn}{shili@nju.edu.cn} \and Jia Ye \footnotemark[\value{footnote}] \\ School of Computer Science, \\ Nanjing University, \\ Nanjing, Jiangsu Province, China \\\href{mailto:jiaye@smail.nju.edu.cn}{jiaye@smail.nju.edu.cn}}
\date{}

\maketitle

\begin{abstract}
We revisit the simultaneous approximation model for the correlation clustering problem introduced by Davies, Moseley, and Newman~\cite{davies2023one}. The objective is to find a clustering that minimizes given norms of the disagreement vector over all vertices. 

We present an efficient algorithm that produces a clustering that is simultaneously a $63.3$-approximation for all monotone symmetric norms. This significantly improves upon the previous approximation ratio of $6348$ due to Davies, Moseley, and Newman~\cite{davies2023one}, which works only for $\ell_p$-norms. 

To achieve this result, we first reduce the problem to approximating all top-$k$ norms simultaneously, using the connection between monotone symmetric norms and top-$k$ norms established by Chakrabarty and Swamy \cite{chakrabarty2019approximation}. Then we develop a novel procedure that constructs a $12.66$-approximate fractional clustering for all top-$k$ norms. Our $63.3$-approximation ratio is obtained by combining this with the $5$-approximate rounding algorithm by Kalhan, Makarychev, and Zhou~\cite{kalhan2019correlation}.

We then demonstrate that with a loss of $\epsilon$ in the approximation ratio, the algorithm can be adapted to run in nearly linear time and in the MPC (massively parallel computation) model with poly-logarithmic number of rounds. 

By allowing a further trade-off in the approximation ratio to $(359+\epsilon)$, the number of MPC rounds can be reduced to a constant.
\end{abstract}

\thispagestyle{empty} 
\clearpage

\tableofcontents
\clearpage

\section{Introduction}
Clustering is a classic problem in unsupervised machine learning. It aims to classify a given set of data elements based on their similarities, with the goal of maximizing the similarity between elements within the same class and minimizing the similarity between elements in different classes. Among the various graph clustering problems, correlation clustering stands out as a classic model. Initially proposed by Bansal, Blum, and Chawla \cite{BBC04}, the model has numerous practical applications, including automated labelling \cite{chakrabarti2008graph,agrawal2009generating}, community detection and mining \cite{chen2012clustering,DBLP:conf/www/VeldtGW18,shi2021scalable}, and disambiguation task \cite{kalashnikov2008web}, among others.

The input of the standard correlation clustering problem is a complete graph over a set $V$ of $n$ vertices, where edges are partitioned into  the set $E^+$  of $+$edges and the set $E^-$ of $-$edges.  %signed graph $G=(V,E)$, where $E^+ = E$ and $E^- =\binom{V}{2} \setminus E$. Here, $E^+$ represents the set of $+$edges and $E^-$ represents the set of $-$edges. For simplicity, we set $n = |V|$ and $m = |E^+|$. 
The output of the problem is a clustering  (or a partition) $\calC$ of $V$ that minimizes the number of edges in disagreement: an edge $uv \in {V \choose 2}$ is in disagreement if $uv \in E^+$ but $u$ and $v$ are in different clusters in $\calC$, or $uv \in E^-$ but $u$ and $v$ are in a same cluster in $\calC$.  Throughout the paper, we shall use a graph $G = (V, E)$ to denote a correlation clustering instance, with $E$ being $E^+$ and ${V \choose 2} \setminus E$ being $E^-$. %Let $m = |E|$ be the number of $+$edges. Shi: The definition of $m$ is not needed here.

This problem is known to be APX-Hard \cite{CGW05}. There has been a long stream of $O(1)$-approximation algorithms for the problem \cite{BBC04,CGW05,ACN08,CMSY15,CLN22,CLLN23, cao2024understanding}, with the current best approximation ratio being $1.437$~\cite{cao2024understanding}. In the same paper, the authors presented an improved hardness of 24/23 for the problem, which also made the constant explicit.

Besides the standard setting, other objectives have been studied recently, with the goal of minimizing some norm of the \textit{disagreement vector} of the clustering $\calC$ over vertices. For a clustering $\mathcal{C}$ of $V$, the disagreement vector of $\calC$ is defined as $\cost_{\calC} \in \mathbb{Z}_{\ge 0}^{n}$, where $\cost_{\calC}(u)$ for every $u \in V$ is the number of edges incident to $u$ that are in disagreement with respect to $\mathcal{C}$. Given some norm $f:\R_{\geq 0}^n \to \R_{\geq 0}$ \footnote{This means $f$ satisfies $f(\alpha x) = \alpha f(x)$ for every real $\alpha \geq 0$ and $x \in \R_{\geq 0}^n$, and $f(x + y) \leq f(x) + f(y)$ for every $x, y \in \R_{\geq 0}^n$}, the goal of the problem is to minimize $f(\cost_{\calC})$. Notice that the standard correlation clustering problem corresponds to the case where $f$ is the $\ell_1$ norm.

Puleo and Milenkovic~\cite{puleo2016correlation} initiated the study of correlation clustering with the goal of minimizing the $\ell_p$ norm of the disagreement vector, where $p \in [1, \infty]$. They proved that the problem is NP-hard for the $\ell_\infty$-norm objective. Given a fixed $p \in [1, \infty]$, for the $\ell_p$-norms objective, they gave a $48$-approximation algorithm. The approximation ratio was subsequently improved by Charikar, Gupta, and Schwartz \cite{charikar2017local} to $7$ for the $\ell_\infty$-norm, and by Kalhan, Makarychev and Zhou \cite{kalhan2019correlation} to $5$ for the $\ell_p$-norm with any fixed $p \in [1, \infty]$. Very recently, Heidrich, Irmai, and Andres~\cite{heidrich20244} improved the approximate ratio to 4 for the $\ell_\infty$-norm.

Davies, Moseley and Newman \cite{davies2023one} introduced the concept of simultaneous approximation for all $\ell_p$-norms. They developed an efficient algorithm that outputs a single clustering $\calC$, which is simultaneously an $O(1)$-approximation for the $\ell_p$ norm for all $p \in [1, \infty]$. This is rather surprising, as it was not known a priori whether such a clustering $\calC$ even exists. To achieve the goal, they first construct a fractional clustering $x$ that is simultaneously an $O(1)$-approximation for all $\ell_p$ norms and then use the $5$-approximate rounding algorithm of Kalhan, Makarychev, and Zhou~\cite{kalhan2019correlation} to round $x$ into an integral clustering $\calC$.  Crucially, the algorithm of \cite{kalhan2019correlation}  guarantees a per-vertex $5$-approximation, meaning that $\cost_\calC(u)$ is at most $5$ times the fractional number of edges in disagreement incident to $u$, for every $u \in V$.  This strong property is necessary to obtain the final simultaneous $O(1)$-approximation in \cite{davies2023one}.

% MPC model
In light of the growing networks, it is imperative to develop efficient parallel algorithms. % to address this problem. 
This urgency is particularly pronounced in machine learning and data mining applications, where timely and efficient processing is essential for extracting meaningful insights from vast datasets.
Many works in the literature aim to design efficient parallel algorithms \cite{blelloch2012greedy, chierichetti2014correlation, DBLP:conf/nips/PanPORRJ15, fischer2019tight, DBLP:conf/wdag/CambusCMU21, cohen2021correlation, DBLP:conf/innovations/Assadi022, cambus20243+, cao2024breaking}. 
The MPC model, as a theoretical abstraction of several real-world parallel models such as MapReduce \cite{dean2008mapreduce}, is a prevalent methodology employed in these works.
% This problem has also been studied in other settings, including online settings \cite{mathieu2010online, NEURIPS2021_250dd568, CLMP22}, streaming settings \cite{pmlr-v37-ahn15, DBLP:conf/soda/BehnezhadCMT23, makarychev2024single, CambusKLPU-SODA24}, settings with fair or local guarantees \cite{charikar2017local, kalhan2019correlation, ahmadian2023improved, heidrich20234} and settings with differential privacy guarantees \cite{bun2021differentially, Daogao2022, DBLP:journals/corr/abs-2203-01440}. 

% \nairen{While Davies, Moseley, and Newman \cite{davies2023one} achieve an $O(1)$ approximation ratio simultaneously for all $l_p$ norms through a sequential algorithm, their algorithm is nearly linear when the graph is sparse. A natural question arises: Can we design a fast parallel algorithm that achieves simultaneous $O(1)$ approximation for all $l_p$ norms with nearly linear-time work? We answer this question affirmatively.}

\subsection{Our results}
In this paper, we revisit and generalize the simultaneous approximation model for the correlation clustering that was introduced by \cite{davies2023one}. Instead of considering only $\ell_p$ norms, we consider all \emph{monotone symmetric norms}. We say a norm $f:\R_{\geq 0}^n \to \R_{\geq 0}$ is monotone if for every $x, y \in \R_{\geq 0}^n$ with $x \leq y$, we have $f(x) \leq f(y)$. We say $f$ is symmetric if $f(x) = f(x')$ for every $x, x' \in \R_{\geq 0}^n$ such that $x'$ is a permutation of $x$. Such norms were considered in \cite{chakrabarty2019approximation} in the context of load balancing and clustering.  Our first result is that there exists simultaneous $O(1)$-approximation for all monotone symmetric norms for correlation clustering and it can be constructed in polynomial time.
\begin{definition}
    Given a correlation clustering instance $G = (V, E)$ and $\alpha \geq 1$, we say a clustering $\calC$ over $V$ is simultaneously $\alpha$-approximate, or a simultaneous $\alpha$-approximation, for a family $F$ of norms, if we have $f(\cost_\calC) \leq \alpha \cdot f(\cost_{\OPT_f})$ for every $f \in F$, where $\OPT_f$ is the optimum clustering for $G$ under norm $f$.
\end{definition}

\begin{theorem}
\label{thm:sequentialAlgorithm}
    Given a correlation clustering instance $G = (V, E)$, in polynomial time we can construct a simultaneous $63.3$-approximate clustering $\calC$ for the family of monotone symmetric norms.
\end{theorem}

Next, we are concerned with the running time of the algorithm and its implementation under the MPC model. To state the result, we need a formal description of the MPC model. 

\paragraph{The MPC model.} In the MPC model, data is distributed across a set of machines, and computation proceeds in synchronous rounds. During each round, each machine first receives messages from other machines, then performs computations based on this information and its own allocated memory, and finally sends messages to other machines to be received at the start of the next round. Each machine has limited local memory, restricting the total number of messages it can receive or send in a round. %The time complexity is defined by the total number of rounds needed to solve the problem.
The efficiency of the algorithm is measured by the number of rounds, the memory used by each machine, the total memory used by all machines, and the running time over all machines, also known as the total work.

In this paper, we consider the MPC model in the \emph{strictly sublinear regime}: Each machine has $O(n^\delta)$ local memory, where $n$ is the input size and $\delta > 0$ is a constant that can be made arbitrarily small. Under this model, we assume the input received by each machine has size $O(n^\delta)$. 
%I removed O(m) space. m is not defined at this point. 

We then describe the correlation clustering problem under the MPC model in the strictly sublinear regime. We use $n = |V|$ and $m = |E|$ to denote the number of vertices and edges respectively in the input graph $G = (V, E)$. The edges $E$ are distributed across the machines, where each machine has $O(n^\delta)$ memory for a constant $\delta > 0$ which can be made arbitrarily small. At the end of the algorithm, each machine needs to store in its local memory the IDs of the clusters for all the vertices incident to its assigned edges.

Our main result regarding MPC algorithm is given as follows,

%TODO: use the constant that is better than 63.3.
%This paper is focused on optimizing all $l_p$-norms$(p\geq 1)$ for correlation clustering at the same time. The main result of the paper is as follows: \snote{We need to state the result for all monotone symmetric norms, not just $l_p$ norms.}

\begin{theorem}
\label{thm:mainthmlp}
    Let $\epsilon \in (0, 1)$. 
%For any constant $\delta> 0$, and a small enough constant $\epsilon$, 
    There exists a randomized MPC algorithm in the strictly sublinear regime that, given a correlation clustering instance $G = (V, E)$, in $O(\log^3 n)$ rounds outputs a simultaneous $(63.3 + O(\epsilon))$ clustering for $G$ for all monotone symmetric norms. This algorithm succeeds with high probability. It uses %$O(n^\delta)$ memory per machine, 
    $\tilde{O}(m / \epsilon^6)$ total memory and $\tilde{O}(m / \epsilon^6)$ total work.\footnote{As usual, we use $\tilde O(\cdot)$ to hide a poly-logarithmic factor in the input size.}
\end{theorem}

%Then we show that the clustering $\calC$ can be constructed in nearly-linear time with high probability, with a loss of $1+\epsilon$ in the approximation ratio. Indeed, the algorithm can be implemented in the MPC model with poly-logarithmic rounds. 

In particular, the algorithm can be converted into a nearly linear time algorithm that with high probability outputs a $(63.3+O(\epsilon))$-simultaneous approximation for all monotone symmetric norms. \smallskip

Along the way, we develop an MPC rounding algorithm with a per-vertex $(5 + 55\epsilon)$ approximation guarantee, based on the sequential algorithm due to \cite{kalhan2019correlation}. Given its potential independent interest, we state it here for future references.
\begin{restatable}{theorem}{thmroundingmain}
\label{thm:roundingmaintheorem}
Let $\epsilon \in (0,1)$ be a constant. Given a graph $G = (V, E)$ and a set of LP value $( x_{uv} )_{u,v \in V}$ satisfying the approximate triangle inequality, that is, for any $u,v,w \in V$, we have $x_{uv} + x_{uw} + \epsilon \geq x_{vw}$. Let $y_u = \sum_{uv \in E} x_{uv} + \sum_{uv \in {V \choose 2} \setminus E}(1 - x_{uv})$ be the LP disagreement for node $u$. There exists an MPC algorithm that computes a clustering $\calC$ such that for any node $u$, we have 
\begin{align*}
    \cost_{\calC}(u) \leq (5 + 55\epsilon) y_u.
\end{align*}

This algorithm always succeeds but terminates in $O(\log^3 n / \epsilon)$ rounds with high probability and requires $O(n^\delta)$ memory per machine. Moreover, let $K = E \cup \{ uv \in {V \choose 2} \setminus E \mid x_{uv} < 1 \}$ be the set of $+$edges and $-$edges whose LP value is less than 1. The algorithm uses a total memory of $O(|K|\log n)$ and a total work of $O(|K|\log^3 n /\epsilon)$.
\end{restatable}

% \begin{coro}
%     Let $\epsilon \in (0, 1)$. There is a $\tilde O(m/\epsilon^6)$-time randomized algorithm that, given a correlation clustering instance $G = (V, E)$, outputs a clustering $\calC$, that is a simultaneous $(63.3 + \epsilon)$-approximation for $G$ for all monotone symmetric norms with high probability.
% \end{coro}

The $O(\log^3 n)$ round in the above theorem might not be desirable for many applications. Our next result shows that we can reduce the number of rounds to $O(1)$, albeit with a worse $O(1)$ approximation ratio:
%We also show that by using the pre-clustering procedure in \cite{cohen2021correlation}, 
\begin{restatable}{theorem}{thmconstantMPCAlgorithmForCC} \label{thm:constantMPCAlgorithmForCC}
%For any constant $\delta> 0$, and a small enough constant $\epsilon$, 
Let $\epsilon \in (0, 1)$ be a constant.
There exists a randomized MPC algorithm in the strictly sublinear regime that, given a correlation clustering instance $G = (V, E)$, in $\mathbf{O(1)}$ rounds outputs a clustering that is simultaneously a $(359 + \epsilon)$-approximation, for all monotone symmetric norms. This algorithm succeeds with high probability, %uses $n^\delta$ memory per machine, 
and uses a total memory of $\tilde{O}(m / \epsilon^2)$ and a total work of $\tilde{O}(m / \epsilon^2)$.
\end{restatable}

%In the context of the MPC model, we pose the question of whether there exists a MPC algorithm that can return a single clustering that simultaneously provides an $O(1)$-approximation for all monotone and symmetric norms. 

Overall, relative to \cite{davies2023one}, our algorithms demonstrate the following improvements.
\begin{enumerate}[leftmargin=*]
    \item We generalize the family of norms for the simultaneous approximation from $\ell_p$ norms to all monotone symmetric norms.
    \item We obtain a simpler construction, which leads to a much smaller approximation ratio. Using a result from \cite{chakrabarty2019approximation}, to simultaneously approximate all monotone symmetric norms, it suffices to approximate all top-$k$ norms: the top-$k$ norm of a non-negative vector is the sum of its largest $k$ coordinates. Though being more general mathematically, the top-$k$ norms are more convenient to deal with compared to $\ell_p$ norms. 
    \item We can make our algorithm run in nearly linear time. This is the first nearly-linear time simultaneous $O(1)$-approximation algorithm for the problem, even when we restrict to $\ell_p$ norms. In contrast, the algorithm of \cite{davies2023one} runs in nearly linear time only when the graph $G$ has $O(1)$ maximum degree.
    \item We can make our algorithm run in the MPC model with $O(1)$ rounds. Our work is the first to consider the problem in the MPC model. 
\end{enumerate}

\subsection{Overview of Techniques} We then discuss our techniques for each of our main results. 
\paragraph{Polynomial Time Construction of Simultaneous $O(1)$-Approximation for All Symmetric Norms} By \cite{chakrabarty2019approximation}, we can reduce the problem of approximating all monotone symmetric norms to approximating all top-$k$ norms.  We then construct a fractional solution $x$, which is a metric over $V$ with range $[0, 1]$, such that the fractional disagreement vector for $x$ has top-$k$ norm at most $12.66 \cdot \opt_k$ for any $k \in [n]$, where $\opt_k$ is the cost of the optimum clustering under the top-$k$ norm. Then, we can use the 5-approximate rounding algorithm of KMZ \cite{kalhan2019correlation}, to obtain a simultaneous $63.3$-approximation for all top-$k$ norms.  The KMZ rounding algorithm has two crucial properties that we need: it does not depend on $k$ and it achieves a per-vertex guarantee. 

We elaborate more on how to construct the metric $x: {V \choose 2} \to [0, 1]$. 
%Once we have Lemma \ref{lem:Topk2all}, which offers an effective way to bridge monotone and symmetric norms with the top-k norm, it becomes much more intuitive to address the problem using the top-k norm. Our algorithm adopts the LP rounding framework as outlined in \cite{davies2023one}. Given a graph $G = (V, E)$, let $x_{uv}$ represent the LP value for each pair of nodes. For all top-k norms, the constraints of the LP remain the same, while the objective varies. The constraint we apply is the triangle inequality; for any $u, v, w \in V^3$, we require that $x_{uv} + x_{uw} \geq x_{vw}$. Consequently, we can assign each edge an LP value such that all top-k norm LPs are satisfied. We need only demonstrate that the LP value can be bounded within each top-k norm.
%
A natural idea to assign the LP values, that was used by~\cite{davies2023one}, is to set $x_{uv}$ based on the intersection of the neighborhood between $u$ and $v$. Intuitively, the more common neighbors two nodes share, the closer they should be.  A straightforward approach to implementing this idea is to set $x_{uv} = 1 - \frac{|N(u) \cap N(v)|}{\max(d(u), d(v))}$, where $N(u)$ denotes the neighboring nodes of $u$ in $G$ and $d(u) = |N(u)|$ denotes the degree of $u$; it is convenient to assume $u \in N(u)$. This approach works for the top-$1$ norm (i.e., the $\ell_\infty$ norm) as discussed in \cite{davies2023fast}, but fails for the top-$n$ norm (i.e., the $\ell_1$ norm). Consider a star graph, where the optimal clustering under the top-$n$ norm has a cost of $n - 2$. This approach will assign $x_{uv} = 1 - \frac{1}{2} = 1/2$ for all $-$edges, leading to an LP cost of $\Theta(n^2)$ and a gap of $\Omega(n)$. \cite{davies2023one} addressed the issue by rounding up LP values to $1$ for $-$edge, if for a given node, its total $-$edges LP disagreement is larger than the number of its $+$edges. After the transformation, the triangle inequalities are only satisfied approximately, but this can be handled with $O(1)$ loss in the approximation ratio. 

We address this issue using a different approach, that is 
%Rather than directly rounding up these edges, we explore an alternative method 
inspired by the pre-clustering technique in \cite{cohen2021correlation}. We first preprocess the graph $G$ by removing edges $uv \in E$ for which $|N(u) \cap N(v)|$ is small compared to $\max\{d(u), d(v)\}$. Let the resulting graph be $H$. % the processed graph as $H$. Intuitively, if an edge's endpoints have scant common neighbors, this edge poorly indicates the distance between them. 
We then set our LP values as $x_{uv} = 1 - \frac{|N_H(u) \cap N_H(v)|}{\max\{d(u), d(v)\}}$ if $u \neq v$, where $N_H(u)$ is the set of neighbors of $u$ in $H$.  We show that this solution is a $12.66$-approximation for all top-$k$ norms simultaneously. When compared to \cite{davies2023one}, in addition to the improved approximation ratio, we obtain a considerably simpler analysis.

\paragraph{Implementation of Algorithm in Nearly-Linear Time and in MPC Model} We then proceed to discuss our techniques to improve the running time of the algorithm to nearly-linear. The algorithm contains two parts: the construction of the fractional solution $x$ and the rounding procedure. We discuss the two procedures separately. 

Constructing $x$ in nearly linear time poses several challenges. First, the construction of the subgraph $H$ requires us to identify edges $uv \in E$ with small $|N(u) \cap N(v)|$. Second, we can not explicitly assign $x$ values to all $-$edges. %the number of pairs to which we explicitly assign $x$ values should be $\tilde O(|E|)$.  %While there are only $O(m)$ $+$edges in the graph, there may be as many as $O(n^2)$ $-$edges to consider, necessitating a bound on the $-$edges that we assign $x$ values. 
Finally, to compute $x_{uv}$, we need to compute $|N_H(u) \cap N_H(v)|$. 

The first and third challenges can be addressed through sampling, with an $O(\log n)$ factor loss in the running time. To avoid considering too many $-$edges, we only consider $-$edges with length at most $1-\epsilon$. Consequently, we only need to consider $-$edges whose other endpoints share at least an $\epsilon$ fraction of neighbors with $u$. Given that each neighbor of $u$ in $H$ has degree similar to $u$, we demonstrate that there will be at most $O(d(u) / \epsilon)$ $-$edges to consider for each node $u$. Overall, there will be $\tilde O(m \cdot \poly(1/\epsilon))$ $-$edges for which we need to explicitly assign $x$ values.  Moreover, the nearly-linear time algorithm for the construction of $x$ can be naturally implemented in the MPC model, with $O(1)$ number of rounds.  \smallskip

Then we proceed to the rounding algorithm for $x$. We are explicitly given the $x$ values for $+$edges, and for nearly-linear number of $-$edges. For other $-$edges, their $x$ values are $1$. 
The KMZ algorithm works as follows: in each round, the algorithm selects a node $u$ as the cluster center and then includes a ball with some radius, meaning the algorithm includes all nodes $v$ such that $x_{uv} \leq \textmd{radius}$ into the cluster, removes the clustered nodes, and repeats the process on the remaining nodes. The rounding algorithm can be easily implemented in nearly-linear time using a priority queue structure.  This leads to a nearly-linear time simultaneous $O(1)$-approximation for correlation clustering for all monotone symmetric norms. 

The challenge to implement the algorithm in MPC model is the sequential nature of the algorithm. \cite{kalhan2019correlation} observes that in each round, if we select the nodes that maximize $L(u) = \sum_{x_{uv} \leq r}(r - x_{uv})$ as cluster center, we can effectively bound each node's algorithmic cost, where $r$ is the final ratio. However, choosing a node that maximizes some target inherently makes the process sequential. Our key observation is that, instead of selecting the node that maximizes $L(u)$, we can allow some approximation. This strategy still permits achieving a reasonable approximate ratio with an additional $1 + \epsilon$ overhead while allowing the selection of multiple nodes as cluster centers, thereby parallelizing the rounding process. In each round, there might be several candidate cluster centers with conflicts. To resolve these conflicts, we employ the classical Luby's algorithm~\cite{luby1985simple, chierichetti2014correlation} to find a maximal independent set, ensuring that none of the cluster centers have conflicts.

\paragraph{Organization} We give some 
preliminary remarks %preliminaries
in Section~\ref{sec:prelim}. In Section~\ref{sec:top-k}, we describe our simultaneous $O(1)$-approximation algorithm for correlation clustering for all top-$k$ norms. The reduction from any monotone symmetric norm to top-$k$ norms is deferred to Appendix~\ref{sec:all-norm}. Combining the results leads to a simultaneous $O(1)$-approximation algorithm for all monotone symmetric norms. Then in Section~\ref{sec:nearly-linear} and \ref{sec:MPC-solve-rounding}, we show how we can run the algorithm in the MPC model with nearly linear work. In particular, Section~\ref{sec:nearly-linear} and \ref{sec:MPC-solve-rounding} discuss how to solve the LP and round the LP solution in the MPC model, respectively.  The constant round MPC algorithm is described in Section~\ref{sec:MPC-solve-LP}. Theorem \ref{thm:sequentialAlgorithm}, \ref{thm:mainthmlp}, \ref{thm:roundingmaintheorem} and \ref{thm:constantMPCAlgorithmForCC} are proved in Section \ref{sec:top-k}, \ref{sec:MPC-solve-rounding}, \ref{sec:MPC-solve-rounding} and \ref{sec:MPC-solve-LP} respectively.

\section{Preliminaries}
\label{sec:prelim}
The input to correlation clustering is a complete graph whose edges are partitioned into $+$edges and $-$edges. We shall use the graph $G = (V, E)$ of $+$edges to denote an instance. Let $n = |V|$ and $m = |E|$. For simplicity, we assume $E$ contains all the $n$ self-loops $uu, u \in V$.
% We are given a complete signed graph but only represent the $+$edge graph $G = (V, E)$. According to this convention, 
So, $E$ is the set of $+$edges, and ${V \choose 2} \setminus E$ is the set of $-$edges.  The graph $G$ is fixed in most part of the paper. 
%${V \choose 2} \setminus E = \binom{V}{2} \setminus E$ is the set of $-$edges not explicitly represented. We assume that for each node $u \in V$, $uu \in E$ is a $+$edge. We use $m = |E|$ and $n = |V|$ to represent the number of $+$edges and nodes, respectively. 

For any graph $H = (V_H, E_H)$, and any vertex $v \in V_H$, let $N_H(u) = \{v \in V_H \mid uv \in E_H\}$. For any vertex $u \in V_H$ and any subset $S \subseteq V_H$, we define $d_H(u, S) = \sum_{v \in S} \mathbb{1}(uv \in E_H)$ as the number of edges between $u$ and $S$. We simply use $d_H(u)$ for $d_H(u, V_H)$. When the graph $H$ is the input graph $G$, we omit the subscript. So we use $N(u)$ for $N_G(u)$ and $d(u)$ for $d_G(u)$. Notice that $u \in N(u)$ and $d(u) = |N(u)| \geq 1$ for every $u \in U$.  For the input graph $G = (V, E)$ and any two vertex $u, v \in V$, we define $M_{uv} = \max\{d(u), d(v)\}$ as the maximum degree of $u$ and $v$ for simplicity, as this notion will be frequently used. For any two sets $X$ and $Y$, we denote their symmetric difference by $X \Delta Y$. Algorithms are parameterized by constants $ \beta(0<\beta<1), \lambda(0<\lambda<1)$ that will be determined later. \medskip

A norm on $n$-dimensional non-negative vectors is a function $f:\R_{\geq 0}^n \to \R_{\geq 0}$ satisfying $f(\alpha x) = \alpha f(x)$ for every real $\alpha \geq 0$ and $x \in \R_{\geq 0}^n$, and $f(x + y) \leq f(x) + f(y)$ for every $x, y \in \R_{\geq 0}^n$. We say a norm $f:\R_{\geq 0}^n \to \R_{\geq 0}$ is monotone if for every $x, y \in \R_{\geq 0}^n$ with $x \leq y$, we have $f(x) \leq f(y)$. We say $f$ is symmetric if $f(x) = f(x')$ for every $x, x' \in \R_{\geq 0}^n$ such that $x'$ is a permutation of $x$.
 %For a fixed integer $k\in [n]$, top-$k$ norm of the disagreement vector $\cost_\calC$ is equal to the sum of the $k$ largest coordinates of $\cost_\calC$.
 We say $f$ is the top-$k$ norm for an integer $k \in [n]$ if $f(x)$ is equal to the sum of the $k$ largest coordinates of $x$.
 Chakrabarty and Swamy \cite{chakrabarty2019approximation} showed that any monotone and symmetric norm can be written as the maximum of many ordered norms.  This leads to the following lemma which reduces the monotone-symmetric norms to top-$k$ norms. For completeness, we defer its proof to Appendix~\ref{sec:all-norm}. 
\begin{restatable}{lemma}{lemmatopktolpnorm}
\label{lem:Topk2all}
    For any integer $k\in [n]$, if an algorithm returns a single clustering $\mathcal{C_{\text{ALG}}}$ that is simultaneously a $\rho$-approximation for all top-$k$ norm objectives, then $\mathcal{C_{\text{ALG}}}$ is a $\rho$-approximation for any monotone and symmetric norm $f:\mathbb{R}^n_{\geq 0} \rightarrow \mathbb{R}_+$.
\end{restatable}

For a fixed clustering $\mathcal{C}$, we already defined the disagreement vector of $\mathcal{C}$ as $\cost_\mathcal{C} \in \mathbb{Z}_{\ge 0}^n$, with $\cost_\mathcal{C}(u)$ for every $u \in V$ being the number of edges incident to $u$ that are in disagreement w.r.t $\mathcal{C}$.
%Given a correlation clustering instance $G=(V,E)$, an integer $k$, and a clustering $\mathcal{C}$, for a subset $S \subseteq V$, we denote the top-$k$ value by $\cost^k_\calC(G, S) = \max_{T \subseteq S, |T| = k} \sum_{u \in T} \cost_\calC(u)$. When $S = V$, we abbreviate $\cost^k_\calC(G, S)$ as $\cost^k_\calC(G)$ and omit $G$ when the graph is clear from context. 
Given an integer $k$, and a clustering $\mathcal{C}$, we denote the top-$k$ value by $\cost^k_\calC = \max_{T \subseteq V, |T| = k} \sum_{u \in T} \cost_\calC(u)$. %When $S = V$, we abbreviate $\cost^k_\calC(G, S)$ as $\cost^k_\calC(G)$ and omit $G$ when the graph is clear from context. 
Similarly, for %a graph $G = (V, E)$ and 
any fractional vector $(x_{uv})_{u, v \in V}$,  we denote $\cost_{x}(u) = \sum_{uv \in E} x_{uv} + \sum_{uv \in {V \choose 2} \setminus E}(1 - x_{uv})$ as the disagreement for $u$ with respect to $x$. %Given subset $S \subseteq V$, 
The top-$k$ value of $x$ is defined as $\cost^k_x  = \max_{T \subseteq V, |T| = k} \sum_{u \in T} \cost_x(u)$.
% When $S = V$, we will write $\cost^k_\calC(V)$ and $\cost^k_x(V)$ as $\cost^k_\calC$ and $\cost^k_x$ to represent the top-$k$ objective. \bigskip
%When $S = V$, we will write $\cost^k_x(G, S)$ as $\cost^k_x(G)$ and omit $G$ when the graph is clear from context. 
\medskip

We will use the following theorem from \cite{kalhan2019correlation}:
\begin{theorem}
    \label{thm:KMZ}
    Let $G = (V, E)$ be a correlation clustering instance, and $x \in [0, 1]^{{V \choose 2}}$ be a metric over $V$ with range $[0, 1]$.  %Let $y_u = \sum_{v \in N(u)} x_{uv} + \sum_{v \in V \setminus N(u)} (1 - x_{uv})$ for every $u \in U$. 
    There is a polynomial time algorithm that, given $G$ and $x$, outputs a clustering $\calC$ of $V$ such that $\cost_{\calC}(u) \leq 5\cdot \cost_x(u)$ for every $u \in V$. 
\end{theorem}

We will use the following well-known concentration inequalities.
\begin{theorem}[Chernoff Bound]
\label{thm:chernoff}
Let $X_1, X_2, ..., X_k$ be independent random variables taking values in $\{0, 1 \}$. Let $X = \sum_{i} X_i$ be the sum of these $k$ random variables. Then the following inequalities hold:
\begin{enumerate}[label=(\ref{thm:chernoff}\alph*)]
    \item For any $\epsilon \in (0, 1)$, if $E[X] \leq U$, then $\Pr[X \geq (1+ \epsilon) U] \leq \mathrm{exp}(-\epsilon^2U/3)$. \label{thm:chernoffgeq}
    \item For any $\epsilon \in (0, 1)$, if $E[X] \geq U$, then $\Pr[X \leq (1 - \epsilon) U] \leq  \mathrm{exp}(-\epsilon^2U/2)$. \label{thm:chernoffleq}
\end{enumerate}
\end{theorem}

% \section{The $\Delta$-Symmetric Case}
% \input{SymetricCases}

% \section{The general case}
% \input{GeneralCases}

% \section{The cost of edge deletion}
% \input{EdgeDeletion}

\section{Simultaneous $O(1)$-Approximation for Top-$k$ Norms}
\label{sec:top-k}
In this section, we describe our simultaneous $63.3$-approximation for correlation clustering for all top-$k$ norms. The algorithm described in this section runs in polynomial time. %; in Section~\ref{sec:nearly-linear} we show how to run the algorithm in nearly-linear time, with an additive factor of $\epsilon$ loss in the approximation ratio. 
It first constructs an LP solution to the top-$k$ linear program using a combinatorial procedure. Crucially, the construction does not depend on the value of $k$. We show that the solution has cost $12.66$ times the optimum cost under the top-$k$ norm, for any integer $k \geq 1$. Then we use the rounding algorithm of \cite{kalhan2019correlation} to round the LP solution to an integral one. As it gives a vertex-by-vertex 5-approximation guarantee, this leads to a $63.3$-approximation for the top-$k$ norm for any $k$.

The LP for minimizing the top-$k$ norm of the clustering is given in LP~\eqref{lp:metric-lp}.
\begin{equation}\label{lp:metric-lp}
    \min \qquad \cost^k_x \qquad \text{s.t.}
\end{equation}\vspace*{-30pt}

\noindent
\begin{minipage}[t]{0.4\textwidth}
    \begin{align}
        x_{uv}+x_{uw} \ge x_{vw},\  \forall u,v,w\in V \label{LPC:triangle}%\\[10pt]
    \end{align}
\end{minipage}\hfill
\begin{minipage}[t]{0.3\textwidth}
    \begin{align}
        x_{uv} \in [0,1],\  \forall u,v \in V \label{LPC:non-negative}
    \end{align}
\end{minipage}\hfill
\begin{minipage}[t]{0.25\textwidth}
    \begin{align}
        x_{uu} = 0, \ \forall u \in V \label{LPC:self}
    \end{align}
\end{minipage} \bigskip

In the correspondent integer program, $x_{uv}$ for every $u, v \in V$ indicates if $uv$ is separated or not. %, and $\overline{x}_{uv}$ indicates if $uv$ is in disagreement. 
We view $uv$ as an unordered pair and thus $x_{uv}$ and $x_{vu}$ are the same variable. %$y_u$ for any $u \in V$ denotes the number of incorrect edges incident to $u$. 
So $(x_{uv})_{u, v \in V}$ is a metric with distances in $\{0, 1\}$, which is relaxed to $[0, 1]$ in the linear program. This is captured by constraints \eqref{LPC:triangle}, \eqref{LPC:non-negative} and \eqref{LPC:self}.
Notice that $\cost_x(u)$ for any $u \in V$ is a linear function of the $x$ variables.  The top-$k$ norm of the fractional clustering is defined by $\cost^k_x = \max_{S \subseteq V: |S| = k} \sum_{u \in S} \cost_x(u)$. This could be captured by introducing a variable $z$ and constraints $z \geq \sum_{u \in S}\cost_x(u)$ for any $S \subseteq V$ of size $k$, and setting $z$ to be the objective to minimize. For simplicity, we use the form as described. Despite having an exponential number of constraints, the LP can be solved efficiently as there is a simple separation oracle.  Moreover, we use a combinatorial algorithm to construct a solution $x$, and thus the algorithm does not solve the LP. 

\subsection{Algorithm}
The algorithm for constructing the LP solution $x$ is given in Algorithm~\ref{alg:pre-clusteringlp}. It depends on the parameter $\beta \in (0, 1)$, whose value will be specified later. During the process, we construct a subgraph $H$ by removing any edge $uv \in E$ where $u$ and $v$ have significantly different neighbors. We then set $x_{uv}$ as $1 - \frac{|N_{H}(u) \cap N_{H}(v)|}{M_{uv}}$ if $u\neq v$ and $x_{uv} = 0$ otherwise. Recall that $M_{uv} = \max\{d(u), d(v)\}$ is the maximum degree for any nodes $u$ and $v$ in graph $G$. Intuitively, we treat $+$edges as indicators of whether two nodes belong to the same cluster. The first step is to remove edges that should not be in the same cluster. The second step ensures that the more common neighbors two nodes have, the closer their distance should be. 

\begin{algorithm}[t]
    \caption{Construction of norm-oblivious solution $x$ to metric LP. \\
     \textbf{Input}: Graph $G = (V, E)$\\
\textbf{Output}: $(x_{uv})_{u, v \in V} $ %for each edge $uv \in E$. %A partition of $V$, $\mathcal{S}$.
}
\label{alg:pre-clusteringlp}
    \begin{algorithmic}[1]
    \Function {\textsc{AllNormCC}}{$G = (V, E)$}
    \State let $E_H = \{uv \in E: |N(u) \Delta N(v)| \leq \beta \cdot M_{uv}\}$ and $H = (V, E_H)$
 \label{alg:pre-clusteringlpfirst}
    \State let $x_{uv} \leftarrow 1 - \frac{|N_{H}(u) \cap N_{H}(v)|}{\mathrm{\max}(d(u), d(v))}$ for every $uv \in {V \choose 2}$ and $x_{uu} = 0$ for every $u \in V$
    \EndFunction
\end{algorithmic}    
\end{algorithm}

In the remaining part of this section, we will show 
\begin{lemma}
\label{lemma:boundratiomain}
%For any given graph $G = (V, E)$, 
Let $k$ be any integer in $[n]$. 
Algorithm \ref{alg:pre-clusteringlp} outputs a feasible solution $(x_{uv})_{u, v \in V}$ for ~\eqref{lp:metric-lp} such that for any $k$, % and any $U \subseteq V$ with $|U| = k$, 
we have
\begin{align*}
    \cost^k_x \leq 12.66 \cdot %\cost^k_{\mathcal{C}_{\text{OPTtop-k}}}
    \opt^k, 
\end{align*}
where $\opt^k$ is the cost of the optimum solution under the top-$k$ norm. 
%where $\mathcal{C}_{\text{OPTtop-k}}$ be the optimal correlation clustering solution when using the top-$k$ norm objective.
\end{lemma}

\begin{proof}[Proof of Theorem \ref{thm:sequentialAlgorithm}]
    Once we obtain Lemma \ref{lemma:boundratiomain}, \cite{kalhan2019correlation} provides a rounding algorithm for any feasible solution. Assume the final clustering is $\mathcal{C}_{\text{KMZ}}$. \cite{kalhan2019correlation} ensures that for any node $u$, we have $\cost_{\mathcal{C}_{\text{KMZ}}}(u) \leq 5 \cost_x(u)$, meaning the disagreement for $\mathcal{C}_{\text{KMZ}}$ is at most 5 times the disagreement in the LP solution. We can apply the KMZ rounding algorithm to $x_{uv}$ as output by Lemma \ref{lemma:boundratiomain}. For any integer $k \in [n]$, let $U'$ be the set of $k$ vertices with the largest disagreement values with respect to $\mathcal{C}_{\text{KMZ}}$. Then we have:
    \begin{align*}
        \cost^k_{\mathcal{C}_{\text{KMZ}}} \le 5 \cdot \sum_{u \in U'} \cost_x(u) \le 5 \cdot \cost^k_x \leq 5 \cdot 12.66\cdot \opt^k =  63.3 \cdot \opt^k.
    \end{align*}
    By Lemma \ref{lem:Topk2all}, we know that $\mathcal{C}_{\text{KMZ}}$ is a simultaneous 63.3-approximate clustering for all monotone symmetric norms.
\end{proof}

We will first show that our $x$ is feasible in Section \ref{sec:xismetric}, then we will bound the approximate ratio in Section \ref{sec:boundtopknormcost}.

\subsection{The validity of $x$ to LP~\eqref{lp:metric-lp}}
\label{sec:xismetric}
To show that $x$ is a valid solution to LP~\eqref{lp:metric-lp}, it suffices to prove that it is a metric over $V$ with range $[0, 1]$. %, as we follow the definitions of $\bar x_{uv}$'s and $y_u$'s. 
Moreover, \eqref{LPC:non-negative} and \eqref{LPC:self} hold trivially. Therefore, it remains to show that the triangle inequality (i.e., constraint \eqref{LPC:triangle}) is satisfied: 
\begin{lemma}
\label{lemma:metric}
For any $u,v,w \in V$, we have 
    $x_{uv} + x_{uw} \geq x_{vw}$.
\end{lemma}
\begin{proof}
    We can assume $u, v, w$ are distinct since otherwise the inequality holds trivially. We assume that $d(v) \geq d(w)$ wlog. 
    {%\color{blue}
        \begin{align*}
            x_{uv} + x_{uw} - x_{vw} &\quad=\quad \Big(1-\frac{|N_{H}(u)\cap N_{H}(v)|}{M_{uv}}\Big) + \Big(1-\frac{|N_{H}(u)\cap N_{H}(w)|}{M_{uw}}\Big) - \Big(1-\frac{|N_{H}(v)\cap N_{H}(w)|}{d(v)}\Big) \\
            &\quad=\quad 1 + \frac{|N_{H}(v)\cap N_{H}(w)|}{d(v)} - \frac{|N_{H}(u)\cap N_{H}(v)|}{M_{uv}} - \frac{|N_{H}(u)\cap N_{H}(w)|}{M_{uw}} \\
            &\quad\geq\quad 1 + \frac{|N_{H}(v)\cap N_{H}(w)|}{M_{uv}} - \frac{|N_{H}(u)\cap N_{H}(v)|}{M_{uv}} - \frac{|N_{H}(u)\cap N_{H}(w)|}{M_{uw}} \\ 
            &\quad\geq\quad 1 - \frac{|N_H(u) \setminus N_H(w)|}{M_{uv}} - \frac{|N_{H}(u)\cap N_{H}(w)|}{M_{uw}}\\
            &\quad\geq\quad 1 - \frac{|N_H(u) \setminus N_H(w)| + |N_{H}(u)\cap N_{H}(w)|}{d(u)} \quad = \quad 1 - \frac{|N_H(u)|}{d(u)} \quad\geq \quad 0.
        \end{align*}
        The first inequality used that $d(v) \leq M_{uv}$, the second one follows from $|N_H(u) \cap N_H(v)| - |N_H(v) \cap N_H(w)| \leq |(N_H(u) \cap N_H(v)) \setminus (N_H(v) \cap N_H(w))| = |(N_H(v) \cap (N_H(u) \setminus N_H(w))| \leq |N_H(u) \setminus N_H(w)|$, and the third one used $d(u) \leq M_{uv}$ and $d(u) \leq M_{uw}$. 
    }
\end{proof}
%\snoteinline{It is better to use a sequence of inequalities/equalities in the proof.}

\subsection{Bounding the Top-$k$ Norm Cost of $x$}
\label{sec:boundtopknormcost}
In this section, we compare the top-$k$ norm cost of $x$ to $\opt^k$. 

\paragraph{Notations} 
We fix the integer $k \in [n]$. Let $\mathcal{C}$ be the clustering that minimizes the top-$k$ norm of disagreement vector, but our analysis works for any clustering.  For every $v \in V$, let $C(v)$ be the cluster in $\calC$ that contains $v$. 

For every $u \in V$, let $\cost^+_\calC(u), \cost^-_\calC(u)$ and $\cost_\calC(u)$ respectively be the number of $+$edges, $-$edges and edges incident to $u$ that are in disagreement in the clustering $\calC$. Recall $\cost^{k}_\calC = \max_{S \subseteq V:|S| = k} \sum_{u \in S}\cost_\calC(u)$ is the top-$k$ norm cost of the clustering $\calC$, thus we have $\opt^k = \cost^k_\calC$.

Let $U$ be the set of $k$ vertices $u$ with the largest $\cost_x(u)$ values. So, $\cost^k_x = \sum_{u \in U} \cost_x(u)$. In order to provide a clear demonstration, we divide all of the edges in $G$ into five parts. First, we separate out the parts that are easily constrained by $\cost^k_\calC$. Let $\varphi^+_1$ be the set of $+$edges that are cut in $\calC$, and $\varphi^-_1$ be the set of $-$edges that are not cut in $\calC$. For the remaining $+$edges in $E$ that are not cut in $\calC$, it is necessary to utilize the properties of $+$edges in $E_H$. To this end, let $\varphi^+_2$ be the set of $+$edges in $E\setminus E_H$ that are not cut in $\calC$, and $\varphi^+_3$ be the set of $+$edges in $E_H$ that are not cut in $\calC$. Finally, we define $\varphi^-_2$ as the set of $-$edges that are cut in $\calC$. Formally, we set 
\begin{align*}
    \varphi^+_1 &:= \{uv \mid uv \in E, C(u) \not= C(v) \},\\
    \varphi^+_2 &:= \{uv \mid uv \in E \setminus E_H, C(u) = C(v) \},   \\
    \varphi^+_3 &:= \{uv \mid uv \in E_H, C(u) = C(v) \}, \\
    \varphi^-_1 &:= \{uv \mid uv \in {V \choose 2} \setminus E, C(u) = C(v)\},\\
    \text{and} \quad \varphi^-_2 &:= \{uv \mid uv \in {V \choose 2} \setminus E, C(u) \not= C(v)\}.
\end{align*} %\snote{We need to change the names later. It is weird to use $\{1,3,4,5,6\}$.}

 Notice that $\varphi^+_3$ contains all the self-loops. For every $(i,j) \in \{(1,+),(2,+),(3,+),(1,-),(2,-)\}$ and $u \in V$, we let $\varphi^j_i(u)$ be the set of pairs in $\varphi^j_i$ incident to $u$. We use $\phi^j_i(u) = \{v: uv \in \varphi^j_i(u)\}$ to denote the end-vertices of the edges in $\varphi^j_i(u)$ other than $u$; so $|\phi^j_i(u)| = |\varphi^j_i(u)|$.  We let $y^j_i(u)$ denote the cost of edges in $\varphi^j_i(u)$ in the solution $x$.  For every $(i,j) \in \{(1,+),(2,+),(3,+),(1,-),(2,-)\}$,  we define $f^j_i = \sum_{u \in U}y^j_i(u)$. Therefore, the top-$k$ norm cost of $x$ is $f^+_1 + f^+_2 + f^+_3 + f^-_1 + f^-_2$. 
%  As in the proof of Lemma~\ref{lemma:metric}, $M_{uv} = \max\{d(u), d(v)\}$ for every $u, v \in V$.
 \medskip

With the notations defined, we can proceed to the analysis. Prior to this, several propositions are presented, which will prove useful in the following analysis. We start with the property of edges in $E \setminus E_H$.

\begin{lemma}
    \label{lemma:varphi1}
    For every $uv \in \varphi^+_2$, we have $\cost_\calC(u) + \cost_\calC(v) \geq \beta \cdot M_{uv}$.
\end{lemma}
\begin{proof}
    Since $C(u) = C(v)$, there are at least $|N(u) \Delta N(v)|$ disagreements incident to $u$ and $v$. For every $uv \in \varphi^+_2$, we have $\cost_\calC(u) + \cost_\calC(v) \geq |N(u) \Delta N(v)| \geq \beta \cdot M_{uv}.$
\end{proof}

Then, we show that edges in $E_H$ have similar degrees.

\begin{lemma}
    \label{lemma:Hdegreebound}
    For every $uv \in E_H$, we have $(1-\beta)\cdot d(v) \le d(u) \le \frac{1}{1-\beta}\cdot d(v)$.
\end{lemma}
\begin{proof}
    For every $uv \in E_H$, we have $|N(u) \Delta N(v)| \le \beta \cdot M_{uv}$. Therefore,
    \begin{align*}
        \min\{d(u),d(v)\} \ge |N(u) \cap N(v)| = |N(u) \cup N(v)| - |N(u) \Delta N(v)| \ge (1-\beta) \cdot \max\{d(u),d(v)\},
    \end{align*}
    which gives us $(1-\beta)\cdot d(v) \le d(u) \le \frac{1}{1-\beta}\cdot d(v)$ since $\beta \in (0,1)$.
\end{proof}

As we are about to analyze the top-$k$ norm objective, we can bound the cost in the solution $x$ using coefficients of $\cost_\calC$. This key observation can be formally demonstrated as follows:

\begin{claim}
    \label{claim:using-coefficients}
    Let $c \in \R_{\geq 0}^{n}$ be a vector and $\alpha > 0$ satisfying that $|c|_\infty \leq \alpha$ and $|c|_1 \leq \alpha k$.  Then we have 
    \begin{align*}
        \sum_{r \in V}c(r) \cdot \cost_\calC(r) \leq \alpha \cdot \cost^k_\calC.
    \end{align*}
\end{claim}
\begin{proof}
    We have $|c/\alpha|_1\leq k$ and $|c/\alpha|_\infty \leq 1$. It is well known that $c/\alpha$ is a convex combination of $\{0,1\}$-vectors, each of which has at most $k$ 1's. For each $\{0, 1\}$-vector $d$ in the combination, we have $\sum_{r \in V} d(r) \cost_{\calC}(r) \leq \cost^k_{\calC}$. Taking the convex combination of these inequalities gives $\sum_{r \in V}\frac{c(r)}{\alpha}\cdot \cost_\calC(r) \leq \cost^k_\calC$, which is equivalent to the inequality in the claim.
\end{proof}

From the definitions of $f^+_1$ and $f^-_1$, we can see that it can be constrained by $\cost^k_\calC$.

 \begin{claim}
    $f^+_1 + f^-_1 \leq \cost^k_\calC$. 
 \end{claim}
 \begin{proof}
    \begin{flalign*}
        && f^+_1 + f^-_1 &= \sum_{u \in U}(y^+_1(u) + y^-_1(u)) = \sum_{u \in U, uv \in \varphi^+_1 \cup \varphi^-_1}x_{uv} &&\\
        && &\leq \sum_{u \in U}(|\varphi^+_1(u)|+|\varphi^-_1(u)|) \leq \sum_{u \in U} \cost_\calC(u) = \cost^k_\calC.  && \qedhere
    \end{flalign*}     
 \end{proof}

To bound the cost of the remaining $+$edges, we separately analyze the cost coefficients of vertices in $f^+_2$ and $f^+_3$. 

\begin{lemma}
    \label{lemma:f+2}
    There exists a vector $c^+_2 \in \R_{\geq 0}^{n}$ with the following properties:
    \begin{enumerate}[label=(\ref{lemma:f+2}\alph*)]
        \item \label{property:f+2-cost} $f^+_2 \leq \sum_{r \in V}c^+_2(r)\cdot \cost_\calC(r)$.
        \item \label{property:f+2-c+2-infty} $c^+_2(r) \leq \frac{2}{\beta} \cdot \frac{|\varphi^+_2(r)|}{d(r)}$, for every $r \in V$.
        \item \label{property:f+2-c+2-1} $|c^+_2|_1 \leq \frac2\beta\sum_{u \in U}\frac{|\varphi^+_2(u)|}{d(u)}$.
    \end{enumerate}
\end{lemma}
\begin{proof}
    We bound $f^+_2$ as follows:
    \begin{align*}
        f^+_2 \leq \sum_{u \in U, uv \in \varphi^+_2} 1 \leq
        \sum_{u \in U, uv \in \varphi^+_2} \frac{\cost_{\mathcal{C}}(u)+ \cost_{\mathcal{C}}(v)}{\beta\cdot M_{uv}} =: \sum_{r \in V} c^+_2(r) \cdot \cost_\calC(r). 
    \end{align*}
    The second inequality used Lemma~\ref{lemma:varphi1}.  Therefore, \ref{property:f+2-cost} holds.

    %We then bound $|c|_\infty$ and $|c|_1$ respectively. 
    To show \ref{property:f+2-c+2-infty}, we bound the coefficients for $\cost_\calC(u)$ and $\cost_\calC(v)$ respectively.  If $u \in U$, the coefficient for $\cost_\calC(u)$ is $\sum_{uv \in \varphi^+_2}\frac{1}{\beta \cdot M_{uv}} \leq \frac{|\varphi^+_2(u)|}{\beta \cdot d(u)}$; if $u \notin U$, the coefficient is $0$.  The coefficient for $\cost_\calC(v)$ is $\sum_{u \in U \cap \phi^+_2(v)} \frac{1}{\beta\cdot M_{uv}} \leq \frac{|\varphi^+_2(v)|}{\beta \cdot d(v)}$. Therefore, $c^+_2(r) \leq \frac{2\cdot|\varphi^+_2(r)|}{\beta\cdot d(r)}$.

    To bound $|c^+_2|_1$, we can replace $\cost_\calC(u)$ and $\cost_\calC(v)$ with 1. Then $|c^+_2|_1 = \sum_{u \in U, uv\in\varphi^+_2}\frac{2}{\beta \cdot M_{uv}} \leq \frac2\beta\sum_{u \in U}\frac{|\varphi^+_2(u)|}{d(u)}$. This proves \ref{property:f+2-c+2-1}. 
\end{proof}

Following the analysis of the cost coefficients $c^+_2$ of $f^+_2$, we analyze the cost coefficients of $f^+_3$ in edge set $E_H$.

\begin{lemma}
    \label{lemma:f+3}
    There exists a vector $c^+_3 \in \R_{\geq 0}^{n}$ with the following properties:
    \begin{enumerate}[label=(\ref{lemma:f+3}\alph*)]
        \item \label{property:f+3-cost} $f^+_3 \leq \sum_{r \in V}c^+_3(r)\cdot \cost_\calC(r)$.
        \item \label{property:f+3-c+3-infty} $c^+_3(r) \leq 2\left(\frac{|\varphi^+_2(r)|\cdot |\varphi^+_3(r)|}{\beta\cdot d^2(r)}  + \frac{|\varphi^+_2(r)|}{\beta\cdot d(r)} + \frac{|\varphi^+_3(r)|}{d(r)}\right)$, for every $r \in V$.
        \item \label{property:f+3-c+3-1} $|c^+_3|_1 \leq 2 \sum_{u \in U} \left(\frac{|\varphi^+_2(u)|\cdot|\varphi^+_3(u)|}{\beta \cdot d^2(u)} + \frac{|\varphi^+_3(u)|}{\beta \cdot d(u)} + \frac{|\varphi^+_3(u)|}{d(u)}\right)$.
    \end{enumerate}
\end{lemma}
\begin{proof}
Fix some $uv \in \varphi^+_3 \subseteq E_H$ with $u \neq v$. %By Algorithm \ref{alg:pre-clusteringlp},  we have $|N(u) \Delta N(v)| \leq \beta \cdot \max(d(u), d(v))$, which implies $  (1 - \beta) d(v) \leq d(u)  \leq \frac{d(v)}{1 - \beta}$. 
We let $\tilde u = \arg\max_{w \in \{u, v\}} d(w)$, and $\tilde v$ be the other vertex in $\{u, v\}$. Notice that $d(\tilde u) = M_{uv}$. Then we have
\begin{align}
    x_{uv} &= 1 - \frac{|N_{H}(u) \cap N_{H}(v)|}{d(\tilde u)} = \frac{d(\tilde u) - |N_{H}(u) \cap N_{H}(v)|}{d(\tilde u)}  = \frac{|\varphi^+_1({\tilde u})| + |\varphi^+_2({\tilde u})| + |\varphi^+_3({\tilde u})| - |N_H(\tilde u) \cap N_H(\tilde v)|}{d({\tilde u})} \nonumber\\
    %&\leq \sum_{u \in U, uv \in \varphi^+_3(u)}\frac{|\varphi^+_2({\tilde u})| + \cost^{+}_{\mathcal{C}}({\tilde u})}{d({\tilde u})} + \sum_{u \in U, uv \in \varphi^+_3(u)} \frac{|\varphi^+_3({\tilde u})| - |N_H(u) \cap N_H(v)|}{d({\tilde u})} \\
    %&\leq \frac{|\varphi^+_2({\tilde u})| + \cost^{+}_{\mathcal{C}}({\tilde u})}{d({\tilde u})} + \frac{|\varphi^+_3({\tilde u})| - |N_H({\tilde u}) \cap N_H({\tilde v})|}{d({\tilde u})} \\
    &\leq \frac{\cost^{+}_{\mathcal{C}}({\tilde u}) + |\varphi^+_2({\tilde u})| + |\varphi^+_2({\tilde v})| + \cost^{-}_{\mathcal{C}}({\tilde v})}{d({\tilde u})} \leq \frac{|\varphi^+_2(u)| + |\varphi^+_2(v)| +  \cost_{\mathcal{C}}(u)+ \cost_{\mathcal{C}}(v)}{M_{uv}}. \label{inequ:xuv-for-phi4}
\end{align}
We prove the first inequality in sequence \eqref{inequ:xuv-for-phi4}.
Notice $N_H(\tilde u) \supseteq \phi^+_3(\tilde u)$. Therefore, 
\begin{align*}
    |\varphi^+_3({\tilde u})| - |N_H(\tilde u) \cap N_H(\tilde v)| &\leq |\phi^+_3({\tilde u})| - |\phi^+_3(\tilde u) \cap N_H(\tilde v)| = |\phi^+_3(\tilde u) \setminus N_H(\tilde v)| \leq |\phi^+_2(\tilde v) \cup \phi^-_1(\tilde v)|\\
    &= |\varphi^+_2(\tilde v) \cup \varphi^-_1(\tilde v)|\leq |\varphi^+_2(\tilde v)| + \cost^-_\calC(\tilde v).
\end{align*}
Above, we used that $\phi^+_3(\tilde u) \setminus N_H(\tilde v) \subseteq C(\tilde u) \setminus N_H(\tilde v) = C(\tilde v) \setminus N_H(\tilde v) \subseteq \phi^+_2(\tilde v) \cup \phi^-_1(\tilde v)$.  Also \eqref{inequ:xuv-for-phi4} holds trivially when $u = v$; this holds for every $uv \in \varphi^+_3$.
\medskip

Therefore, 
\begin{align*}
    f^+_3 &= \sum_{u \in U}y^+_3(u) \\
    &\leq \sum_{u \in U, uv \in \varphi^+_3} \frac{|\varphi^+_2(u)| + |\varphi^+_2(v)| +  \cost_{\mathcal{C}}(u)+ \cost_{\mathcal{C}}(v)}{M_{uv}} \\
    &\leq \sum_{\substack{u \in U, uv \in \varphi^+_3, uw \in \varphi^+_2}} \frac{1}{M_{uv}} + \sum_{u \in U, uv \in \varphi^+_3, vw \in \varphi^+_2} \frac{1}{M_{uv}} + \sum_{u \in U, uv \in \varphi^+_3} \frac{\cost_{\mathcal{C}}(u)+ \cost_{\mathcal{C}}(v)}{M_{uv}}\\
    &\leq \sum_{u \in U, uv \in \varphi^+_3, uw \in \varphi^+_2} \frac{\cost_{\mathcal{C}}(u) + \cost_{\mathcal{C}}(w)}{\beta \cdot M_{uv}M_{uw}} + \sum_{\substack{u \in U, uv \in \varphi^+_3, vw \in \varphi^+_2}} \frac{\cost_{\mathcal{C}}(v) + \cost_{\mathcal{C}}(w)}{\beta \cdot M_{uv}M_{vw}} \\
     &\quad + \sum_{u \in U, uv \in \varphi^+_3} \frac{\cost_{\mathcal{C}}(u)+ \cost_{\mathcal{C}}(v)}{M_{uv}}\\
     &=: \sum_{r \in V} c^+_3(r) \cdot \cost_\calC(r). 
\end{align*}
Again, we used Lemma~\ref{lemma:varphi1} twice to prove the last inequality.  Therefore, \ref{property:f+3-cost} holds.

%We treat the above quantity as a linear combination of $\cost_{\mathcal{C}}(u)$. That is, it is $\sum_{r \in V} c(r) \cost_\calC(r)$.  If we can show that $\sum_{r \in V} c(r)$ and $c(r)$ are both bounded, then we can show that $f^+_1$ is bounded by the top-$k$ instance. 

{%\color{blue}
To prove \ref{property:f+3-c+3-infty}, we consider the coefficients for $\cost_{\mathcal{C}}(u), \cost_{\mathcal{C}}(v)$ and $\cost_{\mathcal{C}}(w)$ respectively. For a $u \in U$, the coefficient for $\cost_{\mathcal{C}}(u)$ is 
\begin{align*}
    \sum_{uv \in \varphi^+_3, uw \in \varphi^+_2}\frac{1}{\beta\cdot M_{uv}M_{uw}} + \sum_{uv\in \varphi^+_3}\frac1{M_{uv}} \leq
    \frac{|\varphi^+_3(u)|\cdot|\varphi^+_2(u)|}{\beta d^2(u)} + \frac{|\varphi^+_3(u)|}{d(u)}.
\end{align*} 

The coefficient for $\cost_{\mathcal{C}}(v)$ is
\begin{align*}
    \sum_{u \in U \cap \phi^+_3(v), vw \in \varphi^+_2}\frac{1}{\beta \cdot M_{uv}M_{vw}} + \sum_{u \in U \cap \phi^+_3(v)}\frac1{M_{uv}} 
\leq \frac{|\varphi^+_3(v)|\cdot|\varphi^+_2(v)|}{\beta d^2(v)} + \frac{|\varphi^+_3(v)|}{d(v)}.
\end{align*}

The coefficient for $\cost_{\mathcal{C}}(w)$ is at most 
\begin{align*}
    &\quad\sum_{u \in U \cap \phi^+_2(w), uv \in \varphi^+_3} \frac{1}{\beta \cdot M_{uv} M_{uw}} + \sum_{vw\in\varphi^+_2, u \in U \cap \phi^+_3(v)} \frac{1}{\beta\cdot M_{uv}M_{vw}} \\
    &\leq \sum_{u \in U \cap \phi^+_2(w)}\frac{d(u)}{\beta \cdot d(u)M_{uw}} + \sum_{vw\in\varphi^+_2} \frac{d(v)}{\beta\cdot d(v)M_{vw}} \\
    &=\sum_{u \in U \cap \phi^+_2(w)}\frac{1}{\beta \cdot M_{uw}} + \sum_{vw\in\varphi^+_2} \frac{1}{\beta\cdot M_{vw}} \leq \frac{|\varphi^+_2(w)|}{\beta \cdot d(w)} + \frac{|\varphi^+_2(w)|}{\beta \cdot d(w)} = \frac{2\cdot|\varphi^+_2(w)|}{\beta \cdot d(w)}.
\end{align*}
Therefore, for every $r \in V$, we have $c^+_3(r) \leq 2\left(\frac{|\varphi^+_2(r)|\cdot |\varphi^+_3(r)|}{\beta\cdot d^2(r)} + \frac{|\varphi^+_2(r)|}{\beta\cdot d(r)} + \frac{|\varphi^+_3(r)|}{d(r)} \right)$.

We then bound $|c^+_3|_1$ as follows:
\begin{align*}
|c^+_3|_1&= \sum_{u \in U, uv \in \varphi^+_3, uw \in \varphi^+_2} \frac{2}{\beta \cdot M_{uv}M_{uw}} + \sum_{u \in U, uv \in \varphi^+_3, vw \in \varphi^+_2} \frac{2}{\beta \cdot M_{uv}M_{vw}} + \sum_{u \in U, uv \in \varphi^+_3} \frac{2}{M_{uv}}\\ 
&\leq \sum_{u \in U} \left( \frac{2\cdot|\varphi^+_2(u)|\cdot|\varphi^+_3(u)|}{\beta \cdot d(u) \cdot d(u)} + \frac{2\cdot |\varphi^+_3(u)|}{\beta \cdot d(u)}  + \frac{2\cdot |\varphi^+_3(u)|}{d(u)}  \right)\\
&= 2 \sum_{u \in U} \left(\frac{|\varphi^+_2(u)|\cdot|\varphi^+_3(u)|}{\beta \cdot d^2(u)} + \frac{|\varphi^+_3(u)|}{\beta \cdot d(u)} + \frac{|\varphi^+_3(u)|}{d(u)}\right).
\end{align*}
}
This finishes the proof of \ref{property:f+3-c+3-1} and thus Lemma~\ref{lemma:f+3}.
\end{proof}

With Lemma~\ref{lemma:f+2} and Lemma~\ref{lemma:f+3}, we can then bound $f^+_2 + f^+_3$: 
\begin{lemma}
    $f^+_2 + f^+_3 \leq \frac4\beta \cdot \cost^k_\calC$. 
\end{lemma}
\begin{proof}
    We define $c(r) = c^+_2(r) + c^+_3(r)$ for every $r \in V$. Then,
    %{\color{gray}
    %    \begin{align*}
    %    c(r) \leq \frac{2\cdot|\varphi^+_2(r)|}{\beta\cdot d(r)} + 2\left(\frac{|\varphi^+_2(r)|\cdot |\varphi^+_3(r)|}{\beta\cdot d^2(r)} + \frac{|\varphi^+_3(r)|}{d(r)} + \frac{|\varphi^+_2(r)|}{\beta\cdot d(r)}\right) = 2\left(\frac{|\varphi^+_2(r)|\cdot |\varphi^+_3(r)|}{\beta\cdot d^2(r)} + \frac{|\varphi^+_3(r)|}{d(r)}\right) +  \frac{4\cdot|\varphi^+_2(r)|}{\beta\cdot d(r)}.
    %    \end{align*}
    %    
    %    $\frac{2(1-t)t}{\beta} + 2t + \frac{4(1-t)}{\beta} = \frac{2}{\beta}((1-t)t+2(1-t)+\beta t)$. The maximum is $\frac{4}{\beta}$, achieved when $t = 0$ and $\beta<1$. So, $c(r) \leq \frac4\beta$ for every $r \in V$. \snote{Need to polish.}
    %}
    {%\color{blue}
        \begin{align*}
            c(r) &\leq \frac{2\cdot|\varphi^+_2(r)|}{\beta\cdot d(r)} + 2\left(\frac{|\varphi^+_2(r)|\cdot |\varphi^+_3(r)|}{\beta\cdot d^2(r)}  + \frac{|\varphi^+_2(r)|}{\beta\cdot d(r)} + \frac{|\varphi^+_3(r)|}{d(r)}\right)\\
            &= 2\left(\frac{|\varphi^+_2(r)|\cdot |\varphi^+_3(r)|}{\beta\cdot d^2(r)} + \frac{|\varphi^+_3(r)|}{d(r)} + \frac{2\cdot |\varphi^+_2(r)|}{\beta\cdot d(r)}\right)\\
            &\leq 2\left(\frac{|\varphi^+_3(r)|}{\beta\cdot d(r)} + \frac{|\varphi^+_3(r)|}{\beta \cdot d(r)} + \frac{2\cdot |\varphi^+_2(r)|}{\beta\cdot d(r)}\right)\\
            &\leq \frac{4}{\beta},
        \end{align*}
        where the first inequality holds because \ref{property:f+2-c+2-infty} and \ref{property:f+3-c+3-infty}, the second inequality holds because $|\varphi^+_2(r)| \leq d(r)$ and $\beta < 1$, the last inequality holds because $|\varphi^+_2(r)| + |\varphi^+_3(r)| \leq d(r)$.
    }

    Similarly, we have

        \begin{align*}
            |c|_1 &\leq \frac2\beta\sum_{u \in U}\frac{|\varphi^+_2(u)|}{d(u)} + 2 \sum_{u \in U} \left(\frac{|\varphi^+_2(u)|\cdot|\varphi^+_3(u)|}{\beta \cdot d^2(u)} + \frac{|\varphi^+_3(u)|}{\beta \cdot d(u)} + \frac{|\varphi^+_3(u)|}{d(u)}\right) \\
            &= 2\sum_{u \in U} \left(\frac{|\varphi^+_2(u)|\cdot|\varphi^+_3(u)|}{\beta \cdot d^2(u)} + \frac{|\varphi^+_3(u)|}{\beta \cdot d(u)} + \frac{|\varphi^+_3(u)|}{d(u)} + \frac{|\varphi^+_2(u)|}{\beta \cdot d(u)}\right)\\
            &\leq 2\sum_{u \in U} \left(\frac{|\varphi^+_2(u)|}{\beta \cdot d(u)} + \frac{|\varphi^+_3(u)|}{\beta \cdot d(u)} + \frac{|\varphi^+_3(u)|}{\beta \cdot d(u)} + \frac{|\varphi^+_2(u)|}{\beta \cdot d(u)}\right)\\
            &\leq \frac{4}{\beta} \cdot k,
        \end{align*}
        where the first inequality holds because \ref{property:f+2-c+2-1} and \ref{property:f+3-c+3-1}, the second inequality holds because $\beta < 1$ and for any $u \in U$ there is $|\varphi^+_3(u)| \leq d(u)$, the last inequality holds because  for any $u \in U$ there is $|\varphi^+_2(u)| + |\varphi^+_3(u)| \leq d(u)$.
    %}
    
    The lemma follows from Claim~\ref{claim:using-coefficients}.
\end{proof}

For the cost of the remaining $-$edges, i.e. $f^-_2$, we also analyze the cost coefficients of each vertex.

\begin{lemma}
    \label{lemma:f-2-properties}
    There exists a vector $c^-_2 \in \R_{\geq 0}^{n}$ with the following properties:
    \begin{enumerate}[label=(\ref{lemma:f-2-properties}\alph*)]
        \item \label{property:f-2-cost} $f^-_2 \leq \sum_{r \in V}c^-_2(r)\cdot \cost_\calC(r)$.
        \item \label{property:f-2-c-2-infty} $c^-_2(r) \leq \frac{2}{1-\beta}$, for every $r \in V$.
        \item \label{property:f-2-c-2-1} $|c^-_2|_1 \leq \frac{2k}{1-\beta}$.
    \end{enumerate}
\end{lemma}
\begin{proof}
% Since for any $uv\in E_H$ we have $|N(u)\Delta N(v)|\leq \beta \cdot \max\{d(u), d(v)\}$, we can conclude that $(1-\beta)d(v)\leq d(u)\leq \frac{d(v)}{1-\beta}$.
We have 
\begin{align*}
    f^-_2 & = \sum_{u \in U, uv \in \varphi^-_2} ( 1- x_{uv}) = \sum_{u \in U, uv \in \varphi^-_2}\frac{|N_{H}(u) \cap N_{H}(v)|}{M_{uv}} \leq \sum_{\substack{u \in U, uv \in \varphi^-_2 \\ w \in N_{H}(u) \cap N_{H}(v)}} \frac{1}{M_{uv}}.
    %\quad\leq\quad \sum_{w \in V}\sum_{\substack{u \in N_{H}(w) \cap U \\v\in N_{H}(w) \cap \phi^-_2(u)}} \frac{1}{M_{uv}}.
\end{align*}

Given that $uv \in \varphi^-_2$ indicates that $C(u) \neq C(v)$, we distinguish between two cases: $C(w) = C(u)$ and $C(w) \neq C(u)$. For simplicity, in the summations below, we automatically impose the constraints $u \in U, uv \in \varphi^-_2$ and $uw, vw \in E_H$ when the vertices involved in constraints are defined. Notice that we have $d(v) \ge (1-\beta)d(w)$ from Lemma~\ref{lemma:Hdegreebound}. %Case 1 means $\calC(w) = \calC(u)$, and case 2 means $\calC(w) \neq \calC(u)$.
\begin{align}
    \sum_{u,v,w:C(w) = C(u)}\frac1{M_{uv}} &\leq \sum_{u,w:C(w) = C(u)} \frac{\cost^+_{\calC}(w)}{\max(d(u), (1-\beta) d(w))}. \label{inequ:f-2-eq}\\
    \sum_{u,v,w:C(w) \neq C(u)}\frac1{M_{uv}} &\leq \sum_{u, w:C(w) \neq C(u)}\frac{d(w)}{d(u)} \leq \sum_{u, w:C(w) \neq C(u)}\frac{1}{1-\beta} \leq \sum_{u} \frac{\cost^+_\calC(u)}{1-\beta}. \label{inequ:f-2-neq}
\end{align}
Adding \eqref{inequ:f-2-eq} and \eqref{inequ:f-2-neq}, and using that $\cost^+_\calC(u) \leq \cost_\calC(u)$ for every $u \in V$, we get
\begin{align*}
    f^-_2 \leq \sum_{u \in U,w \in N_H(u):C(w) = C(u)} \frac{\cost_{\calC}(w)}{\max(d(u), (1-\beta) d(w))} + \sum_{u \in U} \frac{\cost_\calC(u)}{1-\beta} =: \sum_{r\in V}c^-_2(r) \cdot \cost_\calC(r).
\end{align*}

Therefore, \ref{property:f-2-cost} holds.

To prove \ref{property:f-2-c-2-infty}, we consider the coefficients for $\cost_\calC(u)$ and $\cost_\calC(w)$ respectively.
The coefficient for $\cost_\calC(u)$ is $\frac{1}{1-\beta}$.
The coefficient for $\cost_\calC(w)$ is
\begin{align*}
    \sum_{u \in U,w \in N_H(u):C(w) = C(u)} \frac{1}{\max(d(u), (1-\beta) d(w))} &\leq \sum_{u \in U,w \in N_H(u):C(w) = C(u)} \frac{1}{(1-\beta) d(w)} \leq \frac{1}{1-\beta}.
\end{align*}
Therefore, for any $r\in V$, we have $c^-_2(r)\le \frac{2}{1-\beta}$.

By replacing $\cost_\calC(u)$ and $\cost_\calC(w)$ with 1, we then finished the proof of \ref{property:f-2-c-2-1} as follows:
\begin{align*}
    |c^-_2|_1 &\leq \sum_{u \in U,w \in N_H(u):C(w) = C(u)} \frac{1}{\max(d(u), (1-\beta) d(w))} + \sum_{u \in U} \frac{1}{1-\beta}\\
    &\leq \sum_{u \in U,w \in N_H(u):C(w) = C(u)} \frac{1}{d(u)} + \sum_{u \in U} \frac{1}{1-\beta}\\
    &\leq \sum_{u \in U} \frac{|N_H(u)|}{d(u)} + \sum_{u \in U} \frac{1}{1-\beta}\\
    &\leq \sum_{u \in U} 1 + \sum_{u \in U} \frac{1}{1-\beta}
    \leq \sum_{u \in U} \frac{2}{1-\beta}
    = \frac{2k}{1-\beta}.  \qedhere
\end{align*}

\end{proof}

With Lemma~\ref{lemma:f-2-properties}, we can then bound $f^-_2$.

\begin{lemma}
    $f^-_2 \leq \frac{2}{1-\beta} \cdot \cost^k_\calC$.
\end{lemma}
\begin{proof}
    By Claim~\ref{claim:using-coefficients} and Lemma~\ref{lemma:f-2-properties}, we have $f^-_2\le \sum_{r\in V}c^-_2(r) \cdot \cost_\calC(r)\le \frac{2}{1-\beta} \cdot \cost^k_\calC$.
\end{proof}

So the final ratio for Algorithm \ref{alg:pre-clusteringlp} is 
\begin{align*}
    %\color{red}
    1 + \frac4\beta + \frac2{1-\beta}.
\end{align*} 
Let $\beta = 0.5858$, then the ratio is at most $12.66$.  This finishes the proof of Lemma~\ref{lemma:boundratiomain}.

%When we set $\beta = 0.5$ for simplicity, we have $\sum_{u \in U} y_u \leq 12.66 \cdot \cost^k_\calC$.

\section{Implementing Algorithm \ref{alg:pre-clusteringlp} in Nearly Linear Time}
\label{sec:nearly-linear}
In this section, we show how to run Algorithm \ref{alg:pre-clusteringlp} approximately in nearly linear time. Indeed, the algorithm can be implemented in MPC model with $O(1)$ rounds. %Furthermore, we can even implement Algorithm \ref{alg:pre-clusteringlp} within $O(1)$ rounds in MPC model. 
More precisely, we will show the following theorem:

\begin{restatable}{theorem}{thmefficientsettingLP} \label{thm:efficientpre-clusteringlp}
Let $\epsilon> 0$ and $\delta > 0$ be small enough constants. Given a graph  $G = (V, E)$, there exists an MPC algorithm that computes a solution $\{ \tilde{x}_{uv} \}_{u,v \in V}$ such that 
\begin{enumerate}
    \item For any integer $k \in [n]$, %let $\mathcal{C}_{\text{OPTtop-k}}$ be the optimal correlation clustering solution when using the top-$k$ norm objective, then 
    we have $\cost^k_{\tilde{x}} \leq (12.66+ \epsilon) \opt^k$.
    \item  For any $u,v,w \in V$, we have $\tilde{x}_{uv} + \tilde{x}_{uw} + \epsilon \geq \tilde{x}_{vw}$. 
\end{enumerate}
The algorithm succeeds with probability at least $1 - 1/n$. Moreover,  the algorithm runs in $O(1)$ rounds, has a total work of $\tilde{O}(m / \epsilon^6)$, requires $O(n^\delta)$ memory per machine and a $\tilde{O}(m / \epsilon^6)$ total memory. 
\end{restatable}

%\subsection{Efficient Algorithm}
%\paragraph{Difficulty of Implementing Algorithm \ref{alg:pre-clusteringlp} in Nearly-Linear Time}

We give the nearly linear time implementation of Algorithm \ref{alg:pre-clusteringlp} in Algorithm~\ref{alg:pre-clusteringlpEfficient}. Line \ref{alg:AllNormCCBySamplingconstructH} constructs the graph $H$ efficiently. Line \ref{alg:AllNormCCBySamplingboundnegativeedgesstart}-\ref{alg:AllNormCCBySamplingboundnegativeedgesend} find the set $K$ of $-$edges we want to assign LP value. For any $-$edge $uv \not\in K$, we will simply set its LP value as $1$. Last, Line \ref{alg:AllNormCCBySamplingsetlpvaluestart} to Line~\ref{alg:AllNormCCBySamplingsetlpvalueend} is to set up $\tilde{x}_{uv}$ satisfying the conditions in~\ref{thm:efficientpre-clusteringlp}. In the remaining of this section, we will discuss these three parts in detail.

\begin{algorithm}[ht!]
\caption{Nearly Efficient Algorithm for Algorithm \ref{alg:pre-clusteringlp}.\\
\textbf{Input}: Graph $G = (V, E), \epsilon\in (0, 1)$\\
\textbf{Output}: $\{ \tilde{x}_{uv} \}_{u,v \in V}$ such that is $(12.66 + 26\epsilon)$-approximate and satisfy approximate triangle inequality. %A partition of $V$, $\mathcal{S}$.
}
\label{alg:pre-clusteringlpEfficient}
\begin{algorithmic}[1]
\Function {\textsc{AllNormCCBySampling}}{$G = (V, E)$}
    \State Construct subgraph $H = (V, E_H \subseteq E)$ using Corollary \ref{coro:constructHbysampling} 
    \label{alg:AllNormCCBySamplingconstructH} \medskip
    %\State $K \leftarrow \emptyset $   
    
    %\Comment{$K$ is the set of $-$edge that we assign $x_{uv}$ value}
    \For{every $u \in V$} \Comment{Compute set $K$}
    \label{alg:AllNormCCBySamplingboundnegativeedgesstart}
    \State $S_u \leftarrow \emptyset$
    \For{every $v \in d_H(u)$}: add $v$ to $S_u$ with probability $\min(\frac{8\log n}{\epsilon d_H(u)}, 1)$
        \EndFor
    %\For{every $v \in S_u$, $vw \in E_H$ and $wu \not\in E$}: add $uw$ to $K$
%    \EndFor
    \EndFor
    \State $K \gets \{uv \notin E: \exists w \in S_u, vw \in E_H\text{ or } \exists w \in S_v, uw \in E_H\}$ \bigskip
    \label{alg:AllNormCCBySamplingboundnegativeedgesend}
    \State $\tau \leftarrow \frac{400\log n}{\epsilon^5}$      \label{alg:AllNormCCBySamplingsetlpvaluestart} \Comment{Compute the $\tilde x_{uv}$ values}
    \For{every $j \in [1, \log n]$}
    \State $S(j) \leftarrow \emptyset$
    \For{$v \in V$}: add $v$ to $S(j)$ with probability $\min(\tau / 2^j, 1)$
    \EndFor
    \For{every $uv \in E \cup K$ with $M_{uv} \in [2^{j-1}, 2^j)$}:
%    \State let $w \in \{u, v\}$ be the node with larger degree $d_G$.
    %\If{$2^{j-1} \leq \max(d(u), d(v)) < 2^{j}$}
    %\State 
    $\tilde{x}_{uv} \leftarrow 1 - \frac{2^j|N_H(u) \cap N_H(v) \cap S(j)|}{\tau \cdot M_{uv}}$
    %\EndIf
    \EndFor
    \EndFor
    \For{every $uv \in E$}: \textbf{if} $\tilde{x}_{uv} \leq \epsilon$ \textbf{then} $\tilde{x}_{uv} \leftarrow 0$ \EndFor
    \label{alg:AllNormCCBySamplingRoundingstart}
    \For{every $uv \in K$}: \textbf{if} $\tilde{x}_{uv} \geq 1 - \epsilon$ \textbf{then} $\tilde{x}_{uv} \leftarrow 1$ \EndFor
    \label{alg:AllNormCCBySamplingRoundingmid}
    \For{every $uv \in {V \choose 2}\setminus (E \cup K)$}: 
        $\tilde{x}_{uv} \leftarrow 1$
    \EndFor
    \label{alg:AllNormCCBySamplingsetlpvalueend}
    \EndFunction
\end{algorithmic}
\end{algorithm}

\subsection{Line~\ref{alg:AllNormCCBySamplingconstructH}: Construction of Subgraph $H$}
To construct graph $H$ in Line~\ref{alg:AllNormCCBySamplingconstructH} of Algorithm \ref{alg:pre-clusteringlpEfficient}, we can use the sampling method in \cite{cohen2021correlation}:

\begin{restatable}{theorem}{thmcnostructH} [Lemma 3.11 of \cite{cohen2021correlation}]
\label{theorem:constructHbysampling}
Let $\delta, \epsilon, \beta \in (0, 1)$ be constants. There exists an MPC algorithm that, given a graph $G = (V, E)$, outputs a ``Yes/No'' answer for every $uv \in E$, such that the following happens with probability at least $1 - \frac{1}{n^6}$.
\begin{itemize}
    \item For every $uv \in E$ with $|N(u) \Delta N(v)| \leq \beta M_{uv}$, algorithm outputs ``Yes'' for $uv$.
    \item For every $uv \in E$ with $|N(u) \Delta N(v)| \geq (1+\epsilon) \beta M_{uv}$, algorithm outputs ``No'' for $uv$. 
\end{itemize}
The algorithm runs in $O(1)$ rounds, with a total work of $\tilde{O}(m/ \epsilon^2)$, $O(n^\delta)$ memory per machine, and $\tilde{O}(m/\epsilon^2)$ total memory.
\end{restatable}

In \cite{cohen2021correlation}, it is required that the algorithm outputs ``Yes'' for $uv \in E$ if $|N(u) \Delta N(v)| \leq 0.8 \beta M_{uv}$ and outputs ``No'' for $uv \in E$ if $|N(u) \Delta N(v)| \geq \beta M_{uv}$. However, we can replace 0.8 with any constant less than 1. Additionally, the algorithm succeeds with probability $1 - \frac{1}{n}$, but we can increase this to $1 - \frac{1}{n^c}$ for any constant $c$. For completeness, we prove Theorem \ref{theorem:constructHbysampling} in Appendix~\ref{sec:constructH}. The theorem leads to the following corollary:
%
% {
% \color{red} This paragraph was accidentally removed due to an network error. But I think we need a better explanation for the proof of the corollary. 
% }
% \begin{coro}
% \label{coro:constructHbysampling}
% For any constant $\delta> 0$, and $\epsilon\in (0, 1)$, given graph $G = (V, E)$, there exists an MPC algorithm that, given a graph $G = (V, E)$, in $O(1)$ rounds outputs a graph $H = (G, E_H \subset E)$ satisfying the following conditions: If we let $x_{uv} = 1 - \frac{|N_H(u) \cap N_H(v)|}{\max(d(u), d(v))}$, then for any $k \in [1, n]$, %let $\mathcal{C}_{\text{OPTtop-k}}$ be the optimal correlation clustering solution using the top-$k$ norm objective, then 
% we have $\cost^k_x \leq (12.66 + 26\epsilon) \opt^k$.
% Moreover, this algorithm succeeds with probability $1 - \frac{1}{n^3}$, uses $n^\delta$ memory per machine, and uses a total memory of $\tilde{O}(m/\epsilon^2)$ and a total work of $\tilde{O}(m/\epsilon^2)$.
% \end{coro}
\begin{coro}
    \label{coro:constructHbysampling}
    Consider the setting in Theorem~\ref{theorem:constructHbysampling}, and let $H = (V, E_H)$ with $E_H$ being the set of edges $uv \in E$ for which the algorithm outputs ``Yes''. Let $x_{uv} = 1 - \frac{|N_H(u) \cap N_H(v)|}{M_{uv}}$ for every $uv \in {V \choose 2}$ and $x_{uu} = 0$ for every $u \in V$. Then for any $k \in [1, n]$, we have $\cost^k_x \leq 12.66(1 + \epsilon) \opt^k$.
\end{coro}

% \subsection{Analysis of Approximate Ratio}
% In this section, we show that Algorithm~\ref{alg:pre-clusteringlpEfficient} incurs an additive factor of $O(\epsilon)$ in the approximation ratio. 

From now on, we define $x_{uv}$ as in Corollary~\ref{coro:constructHbysampling}: $x_{uv} = 1 - \frac{|N_H(u) \cap N_H(v)|}{M_{uv}}$ for every $uv \in {V \choose 2}$ and $x_{uu} = 0$ for every $u \in V$. Notice that Algorithm~\ref{alg:pre-clusteringlpEfficient} does not compute or main $x$ explicitly; it is introduced only for the sake of analysis. 

\subsection{Line~\ref{alg:AllNormCCBySamplingboundnegativeedgesstart}-\ref{alg:AllNormCCBySamplingboundnegativeedgesend}: Construction of $K$} 
After obtaining $H$, there might be $O(n^2)$ $-$edges for which we need to assign $x$ values in Algorithm \ref{alg:pre-clusteringlp}. Fortunately, most of these $-$edges will have $x$ values greater than $1 - \epsilon$, allowing us to simply disregard them. By doing this, we achieve an approximate triangle inequality rather than a strict triangle inequality.  %TODO: if the values do not affect the rounding, then there is no loss. 

The key observation is as follows: for any $-$edge $uv \in {V \choose 2} \setminus E$ with $x_{uv} \leq 1 - \epsilon$, then $\frac{|N_H(u) \cap N_H(v)|}{M_{uv}} = 1 - x_{uv} \geq \epsilon$, which is equivalent to $|N_H(u) \cap N_H(v)| \geq \epsilon M_{uv}$. If we sample $O\left(\frac{\log n}{\epsilon}\right)$ nodes from $N_{H}(u)$,  at least one node is in $N_H(u) \cap N_H(v)$ with high probability. By considering all neighbors of the sampled nodes for $u$, we will cover all $-$edges $uv$ where $x_{uv} \leq 1 - \epsilon$. Additionally, note that for any neighbor $w \in N_H(u)$, their degrees will not differ significantly. Since we only have $O\left(\frac{\log n}{\epsilon}\right)$ nodes to consider, we will only need to consider $O\left(\frac{d(u) \log n}{\epsilon}\right)$ nodes for each node $u$, which is still a nearly linear number of nodes in total. 

% {\color{gray}Need to update. \paragraph{Algorithm Description} Line \ref{alg:AllNormCCBySamplingboundnegativeedgesstart} to Line \ref{alg:AllNormCCBySamplingboundnegativeedgesend} is to find the set $K$ of $-$edges we want to assign LP value. At Line 3, we initialize set $K$ as empty set, then for each node $u$, we try to identify $-$edges $uv \in {V \choose 2} \setminus E$ with $x_{uv} \leq 1 - \epsilon$. For each node $u$, we will sample $O(\log n / \epsilon)$ nodes from $N_H(u)$ at Line 6. At Line 7-8, we add all neighbor of $S_u$ to $K$ if it's a $-$edge. The main lemma regrading Line \ref{alg:AllNormCCBySamplingboundnegativeedgesstart} to Line \ref{alg:AllNormCCBySamplingboundnegativeedgesend} is as follows, }

We let $K$ be the set constructed in the algorithm. 
\begin{lemma}
    \label{lemma:samplingnegativesizebound}
    %Let \( x_{uv} = 1 - \frac{|N_H(u) \cap N_H(v)|}{M_{uv}} \) for \( uv \in \binom{V}{2} \) and \( \epsilon < 1/3 \). A
    $\forall uv \in {V \choose 2}\setminus E \text{ with } x_{uv} \leq 1 - \epsilon, \Pr[uv \in K] \geq 1 - n^{-6}$. 
    %After running Line \ref{alg:AllNormCCBySamplingboundnegativeedgesstart} to Line \ref{alg:AllNormCCBySamplingboundnegativeedgesend}, for any \( uv \in {V \choose 2} \setminus E \) such that \( x_{uv} \leq 1 - \epsilon \), we have \( uv \in K \) with probability at least \( 1 - \frac{1}{n^6} \). 
    Moreover \( |K| = O\left(\frac{m \log n}{\epsilon}\right) \) with probability at least \( 1 - \frac{1}{n^7} \).
\end{lemma}

\begin{proof}
Fix any \( uv \in {V \choose 2} \setminus E \) with \( x_{uv} \leq 1 - \epsilon \). We have \( |N_H(u) \cap N_H(v)| \geq \epsilon M_{uv} \geq \epsilon \max(d_H(u), d_H(v)) \geq \epsilon d_H(u) \). Let $A_{u, v}$  be the event that at least one node from \( N_H(u) \cap N_H(v) \) is added to \( S_u \). If \( d_H(u) \leq \frac{8\log n}{\epsilon} \), then $\Pr[A_{u, v}] = 1$. Otherwise, we have
\begin{align*}
    \Pr[A_{u, v}] &\geq 1 - \left( 1 - \frac{8\log n}{\epsilon d_H(u)} \right)^{\epsilon d_H(u)} 
    \geq 1 - e^{-8\log n} 
    \geq 1 - \frac{1}{n^8}.
\end{align*}
%Once we add a node \( w \in N_H(u) \cap N_H(v) \) to \( S_u \), the algorithm considers all neighbors of \( w \), adding \( uv \) to \( K \). 
Conditioned on $A_{u,v}$, we have $uv \in K$.
Thus, we have $\Pr[uv \in K] \geq 1 - \frac{1}{n^8}$. The first statement follows from union bound over at most $n^2$ edges. 

To bound the size of \( K \), for any node \( u \) with \( d_H(u) \leq \frac{8\log n}{\epsilon} \), we add the whole set \( N_H(u) \) to \( S_u \), resulting in \( |S_u| \leq \frac{8\log n}{\epsilon} \). On the other hand, if \( d_H(u) > \frac{8\log n}{\epsilon} \), let \( X_{uv} \) be the indicator variable that \( v \) is added to \( S_u \). Then,
\begin{align*}
    \E\left[\sum_{v \in N_H(u)} X_{uv}\right] = d_H(u) \cdot \frac{8\log n}{\epsilon d_H(u)} = \frac{8\log n}{\epsilon}.
\end{align*}
Thus, by the Chernoff bound,
\begin{align*}
    \Pr\left[|S_u| \geq \frac{16\log n}{\epsilon}\right] = \Pr\left[\sum_{v \in N_H(u)} X_{uv} \geq \frac{16\log n}{\epsilon}\right] \leq \exp\left(-\frac{8\log n}{3\epsilon}\right) \leq n^{-8},
\end{align*}
as long as \( \epsilon < 1/3 \). Therefore, with probability at least \( 1 - \frac{1}{n^7} \), we have \( |S_u| \leq \frac{16\log n}{\epsilon} \) for all \( u \in V \). 

Note that for any node \( w \in N_H(u) \), by Lemma \ref{lemma:Hdegreebound}, we have \( d_H(w) \leq d(w) \leq \frac{d(u)}{1 - \beta} \). So, with probability at least $1 - \frac{1}{n^7}$, we have 
\begin{flalign*}
    &&|K| \leq \sum_{u \in V} |S_u| \cdot \frac{d(u)}{1 - \beta} \leq \frac{8\log n}{\epsilon(1 - \beta)} \sum_{u \in V} d(u) = O\left(\frac{m \log n}{\epsilon}\right). && \qedhere
\end{flalign*}
\end{proof}

%\snote{We can not use Lemma \ref{lemma:Hdegreebound} any more. This is a different $H$. }

\subsection{Line~\ref{alg:AllNormCCBySamplingsetlpvaluestart}-\ref{alg:AllNormCCBySamplingsetlpvalueend}: Computing $\tilde{x}_{uv}$ for $E \cup K$}

Once we identify all edges to assign LP values to, we use the sampling method to compute $x_{uv}$, as described in~\cite{davies2023fast}. The key observation is that the size of a given set can be evaluated by sampling $O(\log n)$ elements. Special care is needed when the set is too small; hence, we set $\tilde{x}_{uv}$ to $0$ or $1$ if the size of the set is either too small or too large.

\paragraph{Algorithm Description}
Lines \ref{alg:AllNormCCBySamplingsetlpvaluestart} to \ref{alg:AllNormCCBySamplingsetlpvalueend} describe the setup of $\tilde{x}_{uv}$. Let $\tau = \frac{\log n}{\epsilon^2}$. To evaluate $|N(u) \cap N(v)|$, we define the set $S(j)$ as a subset of nodes obtained by sampling every node in the graph independently with probability $\min(\tau / 2^j, 1)$ (Line 11). We will have $O(\log n)$ different values for $j \in [1, \log n]$. Define $j_u = \max\{j \mid 2^j \leq d(u)\}$ as the maximum integer that is a power of 2 and at most $d(u)$. For any $uv \in E \cup K$, we then set 
\begin{align*}
    \tilde{x}_{uv} = 1 - \frac{ 2^j |N_H(u) \cap N_H(v) \cap S(j_u)|}{ \tau M_{uv}}.
\end{align*}
Additionally, for each $uv \in E$, if $\tilde{x}_{uv} \leq \epsilon$, we set $\tilde{x}_{uv} = 0$. For any $uv \in K$, if $\tilde{x}_{uv} \geq 1 - \epsilon$, we set $\tilde{x}_{uv} = 1$ (Lines \ref{alg:AllNormCCBySamplingRoundingstart}-\ref{alg:AllNormCCBySamplingsetlpvalueend}). Our main lemma regarding the approximate LP value $\tilde{x}_{uv}$ is given below:

\begin{restatable}{lemma}{lemmasmpaling}
\label{lemma:samplingfinalvalue}
%Given a graph $G = (V, E)$ and a subgraph $H = (V, E_H)$, where $E_H \subset E$, let $x_{uv} = 1 - \frac{|N_H(u) \cap N_H(v)|}{\max(d(u), d(v))}$ for any $u, v \in V$. Then 
With probability at least $1 - 1 / n^6$, Algorithm \ref{alg:pre-clusteringlpEfficient} outputs $\tilde{x}_{uv}$ such that
\begin{enumerate}[label=(\ref{lemma:samplingfinalvalue}\alph*)]% [label=(\alph*)]
    \item For any $uv \in E$, we have $\tilde{x}_{uv} \leq (1 + \epsilon) x_{uv}$. For any $uv \in {V \choose 2} \setminus E$, we have $1 - \tilde{x}_{uv} \leq (1 + \epsilon)(1 - x_{uv})$.
    \label{lemma:samplingfinalvaluelpvalue}
    \item For any $u, v, w \in V$, we have $\tilde{x}_{uv} + \tilde{x}_{uw} + 3\epsilon \geq \tilde{x}_{vw}$.
    \label{lemma:samplingfinalvaluetriangle}
\end{enumerate}
\end{restatable}

The proof is similar to the sampling method proof presented in \cite{davies2023fast}. The full proof of Lemma \ref{lemma:samplingfinalvalue} is provided in Appendix \ref{sec:proofofsamplingfinalvalue}.

\subsection{Wrapping up: Proof of Theorem \ref{thm:efficientpre-clusteringlp}}

We can now prove Theorem \ref{thm:efficientpre-clusteringlp}. Once we construct the graph $H$ in Algorithm \ref{alg:pre-clusteringlpEfficient}, by setting $x_{uv} = 1 - \frac{|N_H(u) \cap N_H(v)|}{\max(d(u), d(v))}$, we achieve $\cost^k_x \leq (12.66 + 26\epsilon) \opt^k$. By Lemma \ref{lemma:samplingfinalvalue}, we know that at the end, the algorithm outputs the approximate LP value $\tilde{x}_{uv}$. Therefore, we can bound the cost as $\cost^k_{\tilde{x}} \leq (1 + \epsilon) \cost^k_x \leq (12.66 + 50\epsilon) \opt^k$, where $\opt^k$ is the cost of the optimal correlation clustering solution using the top-$k$ norm objective. The approximate triangle inequality is implied by Lemma \ref{lemma:samplingfinalvaluetriangle}. The final ratio is obtained by scaling $\epsilon$ to $\epsilon / 50$.

The running time of Algorithm \ref{alg:pre-clusteringlpEfficient} is derived from the discussion in this section. Constructing $H$ takes $O(1)$ rounds in the MPC model and $\tilde{O}(m /\epsilon^2)$ total work. By Lemma \ref{lemma:samplingnegativesizebound}, we know that we will only assign at most $\tilde{O}(\frac{m}{\epsilon})$ edges. Since the size of $|N_H(u) \cap N_H(v) \cap S(j_u)|$ is bounded by $\tau = O(\frac{\log n}{\epsilon^5})$, each edge takes $\tilde{O}(1/\epsilon^5)$ time to compute $\tilde{x}_{uv}$. In total, Algorithm \ref{alg:pre-clusteringlpEfficient} takes $\tilde{O}(\frac{m}{\epsilon^6})$ work and total memory. Note that all sampling and for-loops are naturally parallelized, so it only takes $O(1)$ rounds to output $\tilde{x}_{uv}$.

%We now argue that Algorithm~\ref{alg:pre-clusteringlpEfficient} can run in MPC model with $O(1)$ round and nearly-linear total work. Thus, it runs in nearly-linear time in the sequential model. 
%{\color{red} TODO}

\section{Rounding}
\label{sec:MPC-solve-rounding}
We will present a nearly linear time rounding algorithm. Furthermore, our algorithm only takes $\Tilde{O}(1)$ rounds in the MPC model. The purpose of this section is to show 

\thmroundingmain*

We emphasize that even if the LP values satisfy the exact triangle inequality, rather than an approximate triangle inequality, the $\epsilon$ terms in the approximate ratio will still be present. These $\epsilon$ terms arise from two sources: the approximate inequality itself and the inherent characteristics of our MPC algorithm.

Given Theorem \ref{thm:efficientpre-clusteringlp} and Theorem \ref{thm:roundingmaintheorem}, we are now able to show the main result of this paper.

\begin{proof}[Proof of Theorem \ref{thm:mainthmlp}]
We first run Algorithm \ref{alg:pre-clusteringlpEfficient}, which outputs $\tilde{x}_{uv}$ as input to Algorithm \ref{alg:mpcrounding}. Note that by Theorem \ref{thm:efficientpre-clusteringlp}, we know that for any $k \in [1, n]$, we have 
\begin{align*}
    \cost^k_{\tilde{x}} \leq (12.66 + \epsilon) \opt^k
\end{align*}
where $\opt^k$ is the cost of the optimal correlation clustering solution when using the top-$k$ norm objective. By Theorem \ref{thm:roundingmaintheorem}, we know that the final cluster $\calC$ satisfies
\begin{align*}
    \cost^k_\calC \leq (5 + 55\epsilon) \cost^k_{\tilde{x}} \leq (63.3 + O(\epsilon)) \opt^k.
\end{align*}

By Lemma \ref{lem:Topk2all}, we know that $\calC$ is a simultaneous $(63.3 + O(\epsilon))$-approximate clustering for all monotone symmetric norms.

For the number of rounds, Algorithm \ref{alg:pre-clusteringlp} takes $O(1)$ rounds, and Algorithm \ref{alg:mpcrounding} takes $O(\log^3 n / \epsilon)$ rounds, resulting in a total of $O(\log^3 n / \epsilon)$ rounds.

Algorithm \ref{alg:pre-clusteringlp} requires a total memory of $\tilde{O}(m / \epsilon^6)$ and a total work of $\tilde{O}(m / \epsilon^6)$ in the MPC model. By Lemma \ref{lemma:samplingnegativesizebound}, we know that $|K| = O(m \log n / \epsilon)$. Therefore, Algorithm \ref{alg:mpcrounding} requires a total memory of $\tilde{O}(m / \epsilon^2)$ and a total work of $\tilde{O}(m / \epsilon)$ in the MPC model, as established by Theorem \ref{thm:roundingmaintheorem}. In total, it requires a total memory of $\tilde{O}(m / \epsilon^6)$ and a total work of $\tilde{O}(m / \epsilon^6)$ in the MPC model.
\end{proof}
% Define 
% \begin{align*}
%     y_t(w) = \sum_{\substack{u \in V_t, uw \in E }} x_{uw} + \sum_{\substack{ u \in V_t , uw \in {V \choose 2} \setminus E}}( 1 - x_{uw} )
% \end{align*}
% We say a node $w$ is good if $L_t(w) \leq \epsilon_2 y_t(w)$. 

\subsection{Rounding Algorithm}
Assume that we are given an instance graph $G = (V, E)$ and a LP solution $( x_{uv} )_{u,v \in V}$ such that $x_{uv} + x_{uw} + \epsilon \geq x_{vw}$ for any $u, v, w \in V^3$.  Given a subgraph $V_t$ and node $w$, radius $r$, define the ball centering at $w$ with radius $r$ as $\ball_{V_t}(w, r) = \{ u \mid u \in V_t, x_{wu} \leq r \} $. Define  
\begin{align*}
    L_t(w) = \sum_{u \in \ball_{V_t}(w, r) }( r - x_{uw} )
\end{align*}

Note that $L_t(w) \geq r$ since $w$ itself is in $ \ball_{V_t}(w, r)$.

\paragraph{Algorithm Description}Our algorithm works as follows: in each round, we choose a set of cluster centers and choose the ball with radius $\twofive$ as a cluster. More precisely. At step $t$, $V_t$ is the set of unclustered vertices. We first compute $L_t^{\textmd{max}}$ to ensure that for each $u\in V_t$, we have $L_t(u) < (1+\epsilon) L_t^{\textmd{max}}$. For any node $u$, if $L_t(u) \geq L_t^{\textmd{max}}$, we will add $u$ to the set of candidate cluster centers $M_t$~(Line \ref{alg:mpcroundingfindcandidatestart}-\ref{alg:mpcroundingfindcandidateend}). Then, we compute cluster centers by adding vertices in the $M_t$ set to the $S_t$ set with probability $p_t$, where the more vertices in $M_t$ are in $\ball(radius=\twofive)$ with each other, the smaller the probability~(Line \ref{alg:mpcroundingfindclustercenterstart}-\ref{alg:mpcroundingfindclustercenterend}). After that, to avoid conflicts, we remove some cluster centers in $S_t$ if they are too close to each other, and derive the final cluster center set $H_t$~(Line \ref{alg:mpcroundingaviodconflicts}). Let $F_t = \ball_{V_t}(H_t,\twofive)$ be the nodes clustered at step $t$. Then we add each $u \in F_t$ to the cluster from $H_t$ with minimum ID. We will remove the clustered nodes and repeat the above process until all vertices are clustered.

\begin{algorithm}[ht!]
\caption{The Rounding Algorithm.\\
\textbf{Input}: Graph $G = (V, E)$, LP solution $( x_{uv} )_{u, v\in V}$ satisfying $x_{uv} + x_{uw} + \epsilon \geq x_{vw}$ for any $u, v, w \in V^3$.\\
\textbf{Output}: A function of clustering $\calC$; i.e., $C(u)=C(v)$ iff $u$ and $v$ belongs to the same cluster. %A partition of $V$, $\mathcal{S}$.
}
\label{alg:mpcrounding}
\begin{algorithmic}[1]
\Function {\textsc{ParallelRounding}}{$G = (V, E),  (x_{uv}) $}
    \State $V_1 \gets V$.
    \State $C(v)\gets \emptyset$ for all $v\in V$.
    \State $t\gets 1$.
    \While{$V_{t}\neq\emptyset$}
        \State $M_t \gets \emptyset, S_t\gets \emptyset$
        \label{alg:mpcroundingfindcandidatestart}
        \State $L_t^{\textmd{max}} = \textmd{max}\{ r (1+\epsilon)^j | r (1+\epsilon)^j \leq \textmd{max}_{w \in V_t}{L_t(w)}, \text{where $j$ is an integer} \}$ \label{alg:randomizedPIVOTLvalue}
        \For{$u \in V_{t}$ } \Comment{Find cluster center candidate}
            \If{$L_t(u) \ge L_t^{\textmd{max}}$ }.
            \State $M_{t} \gets M_{t} \cup \{ u \}$ 
            \EndIf
        \EndFor
        \label{alg:mpcroundingfindcandidateend}
        \State $\Delta_t \gets \textmd{max}_{u \in M_t}{|\ball_{V_t}(u, \twofive) \cap M_t|}$.\label{alg:mpcroundingfindclustercenterstart}
        \State $p_t \gets 1 / (2\Delta_t)$.
        %\State Let $S_t \gets \emptyset, H_t \gets \emptyset$
        \For{$u \in M_{t}$ }
        \State Add $u$ to $S_t$ with probability $p_t$ \Comment{Choose cluster centers in parallel.}
        \EndFor
        \label{alg:mpcroundingfindclustercenterend}
        \State $H_t = \{u \in S_t \mid \nexists\,v \in S_t \cap V_t \textmd{ such that } x_{uv} \leq \twofive \}$ \Comment{remove cluster centers that conflict.}
        \label{alg:mpcroundingaviodconflicts}
        \State $F_t \gets H_t \cup \ball_{V_t}(H_t, \twofive)$ \Comment{the set of settled vertices in round $t$.}
        \For{$u \in F_t$ }
            \State Let $v$ be the vertex with minimum ID among $\ball_{V_t}(u, \twofive) \cap H_t$.
            \State $C(u) \gets v$ \Comment{adding $u$ to the cluster of $v$.}
        \EndFor
        \State $V_{t + 1} \gets V_t \setminus F_t$, $t\gets t+1$
    \EndWhile
    \State \Return $\calC$
    \EndFunction
\end{algorithmic}
\end{algorithm}

\subsection{Approximate Ratio}
\paragraph{Analysis Framework} We will follow the proof of the sequential rounding algorithm~\cite{kalhan2019correlation}, which is also used in \cite{davies2023fast} for approximate triangle inequality. However, we suffer different difficulties that each round we choose multiple nodes whose $L_t$ value is within $\frac{1}{1 + \epsilon}$ times the maximum $L_t$ value. Let $LP(uv)$ be the LP cost of the edge $uv: LP(uv) = x_{uv}$ if $uv \in E$ and
$uv: LP(uv) = 1 - x_{uv}$ if $uv \in {V\choose2} \setminus E$. Let $ALG(uv) = \mathbb{1}(uv \textmd{ is in disagreement})$. We can define 
\begin{align*}
    \profit(u) = \sum_{uv \in E} LP(uv) - r \sum_{uv \in E} ALG(uv)
\end{align*}
where $r = \onefive - 2\epsilon$. We want to show that $\profit(u) \geq 0$ for all nodes. Let 
\begin{align*}
    \Delta E_t = \{ uv: u \in F_t \textmd{ or } v \in F_t\}
\end{align*}
be the set of edges removed at step $t$. Next, let 
\begin{align*}
    \profit_t(uv) =  \begin{cases}
            LP(uv) - r \cdot ALG(uv), &uv\in \Delta E_t\\
            0, &otherwise
        \end{cases}
\end{align*}
and the profit at step $t$ is defined as
\begin{align*}
    \profit_t(u) = \sum_{v \in V_t} \profit_t(uv) = \sum_{uv \in \Delta E_t} LP(uv) - r \cdot \sum_{uv \in \Delta E_t} ALG(uv)
\end{align*}
We will later show that $\profit_t(u) \geq 0$ for each $t$. Thus, we have 
\begin{align*}
    \profit(u) = \sum_{t} \profit_t(u) \geq 0
\end{align*}
and we can bound the cost of the algorithm by 
\begin{align*}
    \sum_{uv \in E} ALG(uv) \leq \frac{1}{r} \sum_{uv \in E} LP(uv) \leq (5 + 55 \epsilon) \sum_{uv \in E} LP(uv)
\end{align*}
which gives us the following lemma,
\begin{lemma}
\label{lem:correctnessofRoundnig}
Assume that $\epsilon \leq 1/ 20$. Given graph $G = (V, E)$ and a set of LP value $ (x_{uv})_{u,v \in V} $ such that $x_{uv} + x_{uw} + \epsilon \geq x_{vw}$ for any $u,v,w \in V^3$. Let $y_u = \sum_{uv \in E} x_{uv} + \sum_{uv \in {V \choose 2} \setminus E}(1 - x_{uv})$ be the LP disagreement for node $u$. Algorithm \ref{alg:mpcrounding} outputs a clustering $\calC$ such that for any node $u$,
\begin{align*}
    \cost_\calC(u) \leq (5 + 55\epsilon) y_u
\end{align*}
\end{lemma}

The remaining part of this section is to show $\profit_t(u) \geq 0$. We first show that we can bound the profit of any $-$edge $uv$. Recall that $r = \onefive - 2\epsilon$.

\begin{lemma}[Analogue of 5.3 of \cite{kalhan2019correlation}] 
\label{lem:roundingcostfornegativeedge}
 Let $u, v \in V_t$ and $uv \in {V \choose 2} \setminus E$, then $\profit_t(uv) \geq 0$.
\end{lemma}
\begin{proof}
For any $-$edge $uv \in \Delta E_t$, if $uv$ goes to different clusters, then $\profit_t(uv) \geq 0$ since $ALG(uv) = 0$. Otherwise, $u,v$ will go to the same cluster and there must be a cluster center $w$ such that $x_{wu} \leq \twofive$ and $x_{wv} \leq \twofive$. By the approximate triangle inequality, we have $x_{uv} \leq x_{wu} + x_{wv} + \epsilon \leq \fourfive + \epsilon$. Therefore, the profit of $uv$ is at least $\profit_t(uv) \geq 1 - x_{uv} -r \geq \onefive - \epsilon - r \geq 0$. 
\end{proof}

To show that $\profit_t(u) \geq 0$, we take a different view on adding nodes to the cluster center: We sort cluster centers $H_t$ by the ID in increasing order. Let $w_1, w_2, ...,$ be the cluster centers from $H_t$ after sorting, then we will add nodes $\ball_{V_t}(w_1, \twofive)$ to $w_1$'s cluster, remove the clustered nodes, and repeat with the next center on the unclustered nodes. The whole process terminates once we process all cluster centers. 

Let $w$ be the first cluster center such that $x_{uw} \leq \fourfive$ when we do the above process. If all nodes in $\ball_{V_t}(w, r)$ are still unclustered, \cite{kalhan2019correlation} already shows that the cost for short $+$edges of $u$ can be covered by the profit of $\ball_{V_t}(w, r)$. The difficulty of the proof mainly comes from the fact that some nodes from  $\ball_{V_t}(w, r)$ may have been clustered. We will divide $x_{uw}$ into several cases and show that in all cases, we have $\profit_t(u) \geq 0$.

The first two cases are when $w$ does not exist or $x_{uw} \leq \onefive$, we will show that under these two cases, $\profit_t(uv) \geq 0$ for any edge $uv \in \Delta E_t$. Note that by Lemma \ref{lem:roundingcostfornegativeedge}, we have $\profit(uv) \geq 0$ for any $uv \in {V \choose 2} \setminus E$. Therefore, we only need to consider $+$edges.

\begin{lemma}
\label{lem:firstcaseprofit}
If the above $w$ does not exist, then $\profit_t(uv) \geq 0$.
\end{lemma}
\begin{proof}
Consider an arbitrary $+$edge $uv \in E$, if $v$ is clustered at $t$-round, we know there must be some cluster center $a \in H_t$ such that $x_{va} \leq \twofive$. On the other hand, we know that $x_{ua} \geq \fourfive$, so $x_{uv} \geq x_{ua} - x_{va} - \epsilon \geq \twofive - \epsilon \geq r$. Thus, for any $+$edge such that $uv \in \Delta E_t$, we have $\profit_t(uv) \geq x_{uv} - r \geq 0$.
\end{proof}

\begin{lemma}
\label{lem:secondcaseprofit}
If $x_{uw} \leq \onefive$, then $\profit_t(uv) \geq 0$.
\end{lemma}
\begin{proof}
Consider an arbitrary $+$edge $uv \in E$. If $x_{uv} \geq r$, then $\profit_t(uv) \geq x_{uv} - r \geq 0$. Otherwise, we have $x_{uv} \leq r$. 

We claim that when we consider $w$ in the above clustering process, $v$ is still unclustered. That is because $w$ is the first cluster center such that $x_{uw} \leq \fourfive$, for any cluster center $\hat{w}$ that considered before $w$, we have $x_{\hat{w}v} \geq x_{u\hat{w}} - x_{uv} - \epsilon \geq \fourfive - r - \epsilon > \twofive$ by approximate triangle inequality. At the same time,  $\hat{w}$ only includes the $\ball_{V_t}(\hat{w}, \twofive)$ into its cluster, so $\hat{w}$ will not include $v$ to its cluster. Furthermore, we claim that $w$ will include $v$ into its cluster, that's because $x_{vw} \leq x_{uw} + x_{uv} + \epsilon \leq \onefive + r + \epsilon \leq \twofive$.

On the other hand, $w$ also includes $u$ into its cluster since all previous cluster centers are at least $\fourfive$ away from $u$ and $x_{uw} \leq \onefive$. Combining them together, we know that $uv$ will be in the same cluster and $ALG(uv) = 0$, which implies $\profit_t(uv) \geq x_{uv} \geq 0$.
\end{proof}

By Lemma \ref{lem:firstcaseprofit} and \ref{lem:secondcaseprofit}, we know that when $w$ does not exist or $x_{uw} \leq \onefive$, we have $\profit_t(u) = \sum_{uv \in \Delta E_t} \profit_t(uv) \geq 0$.  For the remaining cases, we will follow the blue-print of the proof in~\cite{kalhan2019correlation}. Note that $\profit_{t}(uv) < 0$ only if $uv \in E$ and $x_{uv} \leq r$. We will charge the negative profit to the profit of $\ball_{V_t}(w,r)$. More precisely, let 

\begin{align*}
    \profit_t(u) = \underbrace{\sum_{v \in \ball_{V_t}(w, r)}\profit_t(uv)}_{P_{\textrm{high}}(u)} + \underbrace{\sum_{v \in V_t \setminus \ball_{V_t}(w, r)}\profit_t(uv)}_{P_{\textrm{low}}(u)}
\end{align*}
We will show that $P_{\textrm{high}}(u) \geq (1+\epsilon) L_t(w)$ and $P_{\textrm{low}}(u) \geq -L_t(u)$. Note that Algorithm \ref{alg:mpcrounding} always chooses nodes whose $L_t$ value is within $\frac{1}{1 + \epsilon}$ times of the maximum $L_t$ value as cluster center, so the profit $P_{\textrm{low}}(u)$ can be covered by the profit of $\ball_{V_t}(w, r)$, which is at least $P_{\textrm{high}}(u)$. The following two lemmas connects $P_{\textrm{high}}(u)$ with $(1+\epsilon) L_t(w)$. We again have two different cases for $x_{uw}$ and show in each case, we have $\profit_t(uv) \geq (1+\epsilon)(r - x_{vw})$.
\begin{lemma}
\label{lem:thirdcaseprofit}
If $x_{uw} \in (\twofive, \fourfive]$, for any $v \in \ball_{V_t}(w, r)$, we have $\profit_{t}(uv) \geq (1 + \epsilon)(r - x_{vw})$.
\end{lemma}
\begin{proof}
If $uv \in E$, we have 
\begin{align*}
    x_{uv} &\geq x_{uw} - x_{vw} - \epsilon \geq \twofive - \epsilon - x_{vw}.
\end{align*}
Note that $v$ will be clustered into some node since $v \in \ball_{V_t}(w, r)$, so $uv \in \Delta E_t$ and the profit of $uv$ is at least 
\begin{align*}
    \profit_{t}(uv) &\geq x_{uv} - r \geq \onefive + \epsilon - x_{vw} \\
    &\geq (1 + \epsilon)\onefive - (1 + \epsilon)x_{vw} \\
    &\geq (1 + \epsilon)(r - x_{vw}).
\end{align*}
On the other side, if $uv \in {V \choose 2} \setminus E$, we have 
\begin{align*}
    x_{uv} &\leq x_{uw} + x_{vw} + \epsilon \leq \fourfive + \epsilon + x_{vw}.
\end{align*}
Since $v \in \ball_{V_t}(w, \twofive)$, $v$ must be clustered after we consider $w$, while $u$ is still not clustered at this time. Therefore, we know that $uv$ must be in different clusters and $ALG(uv) = 0$ and 
\begin{align*}
    \profit_{t}(uv) &\geq 1 - x_{uv} \geq \onefive - \epsilon - x_{vw} \\
    &\geq (1 + \epsilon)(\onefive - 2\epsilon) - (1 + \epsilon)x_{vw} \\
    &\geq (1 + \epsilon)(r - x_{vw}).
\end{align*}
In either case, we have $\profit_{t}(uv) \geq (1 + \epsilon)(r - x_{vw})$.

\end{proof}

The last case is $x_{uw} \in (\onefive, \twofive]$.
\begin{lemma}
\label{lem:fourthcaseprofit}

If $x_{uw} \in (\onefive, \twofive]$, for any $v \in \ball_{V_t}(w, r)$, we have $\profit_{t}(uv) \geq (1 + \epsilon)(r - x_{vw})$.
\end{lemma}
\begin{proof}
Note that $v \in \ball_{V_t}(w, r)$, so $uv \in \Delta E_t$. If $uv \in {V \choose 2} \setminus E$, we have 
\begin{align*}
    x_{uv} &\leq x_{uw} + x_{vw} + \epsilon \leq \twofive + r + \epsilon \leq \threefive
\end{align*}
Thus, the profit for $uv \in {V \choose 2} \setminus E$ is at least 
\begin{align*}
    \profit_t(uv) &\geq 1 - x_{uv} - r \geq \twofive - r \\
    &\geq \onefive + 2\epsilon \geq (1 + \epsilon)r 
\end{align*}
If $uv \in E$, based on whether $v$ is clustered when we consider $w$, we have two cases. If $v$ has been clustered when we consider $w$, assume that $v$ is included into $\hat{w}$'s cluster, we know that $x_{u\hat{w}} \geq \fourfive$ and $x_{v\hat{w}} \leq \twofive$. By approximate triangle inequality, we have  
\begin{align*}
    x_{uv} \geq x_{u\hat{w}} - x_{v\hat{w}} - \epsilon \geq \fourfive - \twofive - \epsilon \geq \twofive - \epsilon
\end{align*}
and the profit for $uv$ is 
\begin{align*}
    \profit_t(uv) &\geq x_{uv} - r \geq \twofive - \epsilon - (\onefive - 2\epsilon) \\
    &= \onefive + \epsilon \geq (1 + \epsilon)r
\end{align*}
The left case is $v$ is still unclustered when we consider $w$, so $w$ will add $v$ to its cluster since $v \in \ball_{V_t}(w, r)$. At the same time, $w$ also includes $u$ to its cluster because $x_{uw} \leq \twofive$. Therefore, $ALG(uv) = 0$. We have $x_{uv}$ satisfying
\begin{align*}
    x_{uv} &\geq x_{uw} - x_{vw} - \epsilon \geq \onefive - \epsilon - x_{vw} \\
    &\geq (1 + \epsilon)r - x_{vw} \geq (1 + \epsilon)(r - x_{vw})
\end{align*}
and the profit of $uv$ is at least 
\begin{align*}
    \profit_t(uv) \geq x_{uv} \geq (1 + \epsilon)(r - x_{vw})
\end{align*}
In all cases, we have $\profit_t(uv) \geq (1 + \epsilon)(r - x_{vw})$.
\end{proof}

Once we bound the profit for each edge $uv$ such that $v \in \ball_{V_t}(w, r)$, we now are able to lower bound the $P_{\textrm{high}}(u)$.

\begin{lemma}
\label{lem:highprofit}
If $x_{uw} \in (\onefive, \fourfive]$, then $P_{\textrm{high}}(u) \geq (1 + \epsilon) L_t(w)$.
\end{lemma}
\begin{proof}
Based on Lemma \ref{lem:thirdcaseprofit} and Lemma \ref{lem:fourthcaseprofit}, we know that for any $v \in \ball_{V_t}(w, r)$, we have $\profit_{t}(uv) \geq (1 + \epsilon)(r - x_{vw})$, so 
\begin{align*}
    P_{\textrm{high}}(u) = \sum_{v \in \ball_{V_t}(w, r)}\profit_t(uv) \geq (1 + \epsilon) \sum_{v \in \ball_{V_t}(w, r)}(r - x_{vw}) = (1 + \epsilon) L_t(w)
\end{align*}
\end{proof}
We still have to bound the profit of $P_{\textrm{low}}(u)$, this has been shown in \cite{kalhan2019correlation} and \cite{davies2023fast}. We give the corresponding lemma for completeness. 
\begin{lemma}
\label{lem:lowprofit}
If $x_{uw} \in (\onefive, \fourfive]$, then $P_{\textrm{low}}(u) \geq - L_t(u)$.
\end{lemma}
\begin{proof}
    Note that for each $+$edge $uv \in E$, we have 
\begin{align*}
    \profit_t(uv) \geq x_{uv} - r.
\end{align*}
For each $-$edge $uv \in {V \choose 2} \setminus E$, by lemma \ref{lem:roundingcostfornegativeedge}, we know $\profit_t(uv) \geq 0$. Therefore, for any edge $uv$, we have 
\begin{align*}
    \profit_t(uv) \geq \textrm{min}(x_{uv} - r, 0).
\end{align*}
We can compute $P_{\textrm{low}}(u)$ by
\begin{align*}
    P_{\textrm{low}}(u) &= \sum_{v \in V_t \setminus \ball_{V_t}(w, r)}\profit_t(uv) \geq \sum_{v \in V_t \setminus \ball_{V_t}(w, r)} \textrm{min}(x_{uv} - r, 0) \\
    &\geq \sum_{v \in V_t} \textrm{min}(x_{uv} - r, 0) \geq \sum_{v \in \ball_{V_t}(u, r)} x_{uv} - r = -L_t(u).
\end{align*}
\end{proof}

Combining the profit for $P_{\textrm{high}}(u)$ and $P_{\textrm{low}}(u)$, we can deduce that the profit for any node $u$ is non-negative when $ x_{uw} \in (\onefive, \fourfive]$
\begin{lemma}
If $x_{uw} \in (\onefive, \fourfive]$, then $\profit_t(u) \geq 0$.
\end{lemma}
\begin{proof}
    Note that in each round, we have $L_t(u) \leq (1 + \epsilon) L_t(w)$. By Lemma \ref{lem:highprofit} and \ref{lem:lowprofit}, we have 
\begin{align*}
    \profit_t(u) = P_{\textrm{high}}(u) + P_{\textrm{low}}(u) \geq (1 + \epsilon)L_t(w) - L_t(u) \geq 0.
\end{align*}
\end{proof}

\subsection{Running time}

\begin{lemma}
\label{lem:mpcroundingrunningtimebase}
At some round $s$, let $L^{\max}_s$ and $\Delta_s$ be the values we set up at lines \ref{alg:randomizedPIVOTLvalue} and \ref{alg:mpcroundingfindclustercenterstart}. Then, after $O(\log n)$ rounds, at round $e = s + O(\log n)$, with high probability, either $L^{\max}_e < L^{\max}_s$ or $\Delta_e \leq \frac{\Delta_s}{2}$.
\end{lemma}

\begin{proof}
Assume that $L^{\max}_e = L^{\max}_s$. Note that in each round, we remove some nodes from the graph, so $L_t(u)$ never increases, and the candidate set $M_e \subseteq M_s$. Consider any node $u \in M_s$; we will show that, at $e$-th round, with high probability, either $L_e(u) < L^{\max}_s$ or $|\text{Ball}(u, 2r) \cap M_e| \leq \Delta_s / 2$ .

Assume $L_e(u) \geq L^{\max}_s$ and $|\text{Ball}(u, 2r) \cap M_t| > \Delta_s / 2$ both hold. First, we know $u \in M_e$. Secondly, at any round $t \in [s, e]$, if any node $w \in \text{Ball}(u, 2r) \cap M_t$ is chosen as a cluster center and added to $H_t$, then $u$ will be added to $w$'s cluster and will be removed from $V_t$. Third, for any node $w \in \text{Ball}(u, 2r) \cap M_t$, we have $|\text{Ball}(w, 2r) \cap M_t| \leq \Delta_s$.

Let $A_u$ be the event that at the $t$-th round, exactly one node $w \in \text{Ball}(u, 2r) \cap M_t$ is added to $S_t$ and no node $v \in \text{Ball}(w, 2r) \cap M_t$ is added to $S_t$. Then
\[
\Pr[A_u] = \sum_{w \in \text{Ball}(u, 2r) \cap M_t} p_t \cdot (1 - p_t)^{|\text{Ball}(w, 2r) \cap M_t \setminus \{ w \}|} \geq \sum_{w \in \text{Ball}(u, 2r) \cap M_t} p_t \cdot (1 - p_t)^{\Delta_s}.
\]

Note that $\frac{1}{2\Delta_s} \leq p_t \leq \frac{1}{\Delta_s} \leq \frac{1}{2}$, so
\begin{align*}
    \Pr[A_u] \geq \sum_{w \in \text{Ball}(u, 2r) \cap M_t} \frac{1}{2\Delta_s} \cdot \left(1 - \frac{1}{\Delta_s}\right)^{\Delta_s} \geq \sum_{w \in \text{Ball}(u, 2r) \cap M_t} \frac{1}{8\Delta_s}.
\end{align*}

Therefore,
\begin{align*}
    \Pr[A_u] \geq \frac{\Delta_s}{2} \cdot \frac{1}{8\Delta_s} = \frac{1}{16}.
\end{align*}

So, at each round $t \in [s, e]$, with at least $1/16$ probability, $u$ will be added to some cluster and removed. After $O(\log n)$ rounds, for any node $u \in M_s$ such that $L_e(u) \geq L^{\max}_s$ and $|\text{Ball}(u, 2r) \cap M_e| > \Delta_s / 2$, with high probability, $u$ will be removed from the graph. Our argument holds by applying a standard union bound for all nodes from $M_s$.
\end{proof}
The next lemma gives us an upper bound of the number of rounds,

\begin{lemma}
\label{lem:mpcroundingtimelogn}
Algorithm \ref{alg:mpcrounding} terminates after at most $O(\log^3 n / \epsilon)$ rounds, w.h.p.
\end{lemma}
\begin{proof}
Based on Lemma \ref{lem:mpcroundingrunningtimebase}, we know that after $O(\log n)$ rounds, either $L^{\max}_t$ drops by $1 + \epsilon$ or $\Delta_t$ drops by half. If $L^{\max}_t$ does not decrease for $O(\log^2 n)$ rounds,  $\Delta_t$ will become $0$ and all nodes in $M_t$ will be removed. Therefore, $L^{\max}_t$ drops by $1 + \epsilon$ after $O(\log^2 n)$ rounds. After $O(\log^2 n \log_{(1+\epsilon)}n) = O(\log^3 n /\epsilon)$ rounds,  $L^{\max}_t$ become $0$ and all nodes in $G$ will be clustered.
\end{proof}

Now we can prove the main theorem in this section.

\begin{proof}[Proof of Theorem \ref{thm:roundingmaintheorem}]
By Lemma \ref{lem:correctnessofRoundnig}, we know that Algorithm \ref{alg:mpcrounding} returns a solution such that 
\begin{align*}
    \cost_{\calC}(u) \leq (5 + 55\epsilon) y_u.
\end{align*}

By Lemma \ref{lem:mpcroundingtimelogn}, we know that the algorithm terminates in $O(\log^3 n /\epsilon)$ rounds. In each round, Algorithm \ref{alg:mpcrounding} is highly parallelized and only takes $O(1)$ rounds. In total, Algorithm \ref{alg:mpcrounding} takes $O(\log^3 n / \epsilon)$ rounds and $O( |K| \log^3 n /\epsilon)$ total work, where $K$ is the set of $+$edges and $-$edges whose LP value is less than 1.

\end{proof}

\section{A Constant Round MPC Algorithm}
\label{sec:MPC-solve-LP}
We will show Theorem \ref{thm:constantMPCAlgorithmForCC} in this section. We repeat for convenience, 
\thmconstantMPCAlgorithmForCC*

\subsection{Algorithm}
Theorem \ref{thm:mainthmlp} gives us an $O(\log^3 n)$ rounds MPC algorithm. The bottleneck is the rounding procedure. To achieve a constant rounds MPC algorithm, instead of setting up the LP and rounding, we use the pre-clustering algorithm from \cite{cohen2021correlation}, which is very useful for $\ell_1$-norm correlation clustering. We show that the pre-clustering algorithm can also provide an $O(1)$-approximate ratio for all monotone symmetric
norms simultaneously.

\paragraph{Algorithm Description} The algorithm from \cite{cohen2021correlation} is parameterized by $\beta$ and $\lambda$. It has three steps: 
\begin{enumerate}
    \item The first step is the same as the first step of Algorithm \ref{alg:pre-clusteringlp}, where we compute the graph $H$ (Line \ref{alg:pre-clusteringfirst}).
    \item The algorithm marks a node as light if it loses more than a $\lambda$ fraction of its neighbors in the first step. Otherwise, it marks the node as heavy. The algorithm removes all edges between two light nodes in $H$ (Line \ref{alg:pre-clusteringsecondstart} - Line \ref{alg:pre-clusteringsecondend}).
    \item The last step is to output the connected components $F$ of the final graph.
\end{enumerate}

\begin{algorithm}[ht!]
\caption{Pre-clustering – Algorithm 1 in \cite{cohen2021correlation}. \\
\textbf{Input}: Graph $G = (V, E)$, \\
\textbf{Output}: Clustering $F$.}
\label{alg:pre-clustering}
\begin{algorithmic}[1]
\Function {\textsc{PreClustering}}{$G = (V, E)$}
 \State Let $E_H = \{uv \in E: |N(u) \Delta N(v)| \leq \beta \cdot \max(d(u), d(v))\}$ and $H = (V, E_H)$
 \label{alg:pre-clusteringfirst}
\For{$v \in V$} \label{alg:pre-clusteringsecondstart}
\If{$d_H(v) < (1 - \lambda) d(v)$} Mark it as light
\Else \quad Mark it as Heavy
\EndIf
\EndFor  
\State Let $E_{\tilde{G}} = \{ uv \in E_H: u \mathrm{\ or\ } v \mathrm{\ is\ heavy} \}$ and $\tilde{G} = (V, E_{\tilde{G}})$ \label{alg:pre-clusteringsecondend}
\State Compute its connected components on $\tilde{G}$, denoted as $F$,
\Return $F$.
\EndFunction
    \end{algorithmic}    
\end{algorithm}

The main reason we can achieve a constant rounds MPC algorithm is the simplicity of steps 2 and 3. \cite{cohen2021correlation} already showed that Algorithm \ref{alg:pre-clustering} can be implemented within $O(1)$ rounds and is $O(1)$-approximate for correlation clustering under $\ell_1$-norm objective. We extend their proof for approximate ratio and show that Algorithm \ref{alg:pre-clustering} outputs an $O(1)$-approximate clustering $F$ for any top-$k$ norm objective. More concretely, we have the following lemma:

\begin{lemma}
\label{lemma:preclusteringratio}
Assume that $8\beta + \lambda \leq 3/8$. Given any graph $G$, Algorithm \ref{alg:pre-clustering} outputs a clustering $F$ such that for any integer $k \in [n]$, we have 
\begin{align*}
    \cost^k_{F} \leq \left(\frac{3}{\beta} + \frac{1}{\lambda} + \frac{1}{\beta\lambda} + 8\right) \cdot \opt^k
\end{align*}
where $\opt^k$ is the cost of the optimum solution under the top-$k$ norm.
\end{lemma}

\begin{proof}[Proof of Theorem \ref{thm:constantMPCAlgorithmForCC}]
By setting $\beta = 0.0275$ and $\lambda = 0.155$, we achieve an approximate ratio of 359. The approximate ratio is affected by the fact that Theorem \ref{theorem:constructHbysampling} only gives us an approximate estimate of the neighbors. However, similar to Algorithm \ref{alg:pre-clusteringlp}, this only adds an extra $1 + \epsilon$ multiplicative factor to the approximate ratio, resulting in a final ratio of $359 + \epsilon$.

By Lemma \ref{lem:Topk2all}, we know that $F$ is a simultaneous $(359 + \epsilon)$-approximate clustering for all monotone symmetric norms.

By Theorem \ref{theorem:constructHbysampling}, we can implement step 1 in $O(1)$ rounds with $\tilde{O}(m /\epsilon^2)$ total memory and work. Step 2 can be implemented in $O(1)$ rounds with $\tilde{O}(m)$ total memory and work. \cite{cohen2021correlation} showed that computing the connected components on $\tilde{G}$ takes $O(1)$ rounds and $\tilde{O}(m)$ total memory and work. Combining all these steps gives us the target running time and total memory and work.
\end{proof}

\subsection{Approximate Ratio}
The proof follows the blueprint of the proof in \cite{cohen2021correlation} and resembles the proof used to bound the approximate ratio for Algorithm \ref{alg:pre-clusteringlp} in Section \ref{sec:boundtopknormcost}. We will first bound the approximate ratio loss for steps 1 and 2, then we compare the final clustering $F$ with the optimal top-$k$ clustering in the graph where we remove all edges deleted in steps 1 and 2.

\paragraph{Notations} 
We fix the integer $k \in [n]$. $\mathcal{C}$ is the clustering that minimizes the top-$k$ norm of the disagreement vector.  For every $v \in V$, $C(v)$ is the cluster in $\calC$ that contains $v$. Let $U$ be the set of $k$ vertices $u$ with the largest $\cost_F(u)$ values. So, the top-$k$ norm of $F$ is $\sum_{u \in U}\cost_F(u)$.

% Again, for every $u \in V$, let $\cost^+_\calC(u), \cost^-_\calC(u)$ and $\cost_\calC(u)$ respectively be the number of $+$edges, $-$edges and edges incident to $u$ that are in disagreement in the clustering $\calC$.

 Similarly, we divide all $+$edges in $G$ into three parts. First, we separate out the parts easily constrained by $\cost^k_\calC$. Let $\varphi^+_1$ be the set of $+$edges in $E_G$ that are cut in $\calC$. For the remaining $+$edges in $E_G$ that are not cut in $\calC$, we divide them into two categories depending on when we remove the edge. Let $\varphi^+_2$ be the set of $+$edges in $E_G\setminus E_H$ that are not cut in $\calC$, and $\varphi^+_4$ be the set of $+$edges in $E_H\setminus E_{\tilde{G}}$ that are not cut in $\calC$. In other words,  $\varphi^+_2$ and $\varphi^+_4$ are edges that are not cut in $\calC$ and removed at steps 1 and 2. respectively. Formally, we set 
\begin{align*}
    \varphi^+_1 &:= \{uv \mid uv \in E_G, C(u) \not= C(v) \},\\
    \varphi^+_2 &:= \{uv \mid uv \in E_G \setminus E_H, C(u) = C(v) \},   \\
    \varphi^+_4 &:= \{uv \mid uv \in E_H \setminus E_{\tilde{G}}, C(u) = C(v) \}, \\
\end{align*} %\snote{We need to change the names later. It is weird to use $\{1,3,4,5,6\}$.}

For every $i \in \{1, 2, 4\}$ and $u \in V$, we let $\varphi^+_i(u)$ be the set of pairs in $\varphi^+_i$ incident to $u$. 
% We use $\phi^+_i(u) = \{v: uv \in \varphi^+_i(u)\}$ to denote the end-vertices of the edges in $\varphi^+_i(u)$ other than $u$; so $|\phi^+_i(u)| = |\varphi^+_i(u)|$. 
For every $i \in \{1, 2, 4\}$,  we define $g^+_i = \sum_{u \in U}|\varphi^+_i(u)|$. Therefore, when we remove edges from the graph $G$, we will lose at most $ g^+_1 + g^+_2 + g^+_4$ number of edges for nodes from $U$. 

We have $g^+_1 \leq \cost^k_{\mathcal{C}}$. Note that in Lemma \ref{lemma:f+2}, we actually show a stronger argument, $\sum_{u \in U}|\varphi^+_2(u)|$ is bounded by $\frac{2}{\beta}\cost^k_{\mathcal{C}}$. This gave us the following Corollary for $g^+_2$. 

\begin{coro}
    \label{coro:g+2}
    There exists a vector $c^+_2 \in \R_{\geq 0}^{n}$ with the following properties:
    \begin{enumerate}[label=(\ref{coro:g+2}\alph*)]
        \item \label{property:g+2-cost} $g^+_2 \leq \sum_{r \in V}c^+_2(r)\cdot \cost_\calC(r)$.
        \item \label{property:g+2-c+2-infty} $c^+_2(r) \leq \frac{2}{\beta} \cdot \frac{|\varphi^+_2(r)|}{d(r)} \leq \frac{2}{\beta}$, for every $r \in V$.
        \item \label{property:g+2-c+2-1} $|c^+_2|_1 \leq \frac2\beta\sum_{u \in U}\frac{|\varphi^+_2(u)|}{d(u)} \leq \frac{2k}{\beta}$.
    \end{enumerate}
\end{coro}

We still need to bound $g^+_4$.

\begin{lemma}
    \label{lemma:g+3}
    There exists a vector $c^+_4 \in \R_{\geq 0}^{n}$ with the following properties:
    \begin{enumerate}[label=(\ref{lemma:g+3}\alph*)]
        \item \label{property:g+3-cost} $g^+_4 \leq \sum_{r \in V}c^+_4(r)\cdot \cost_\calC(r)$.
        \item \label{property:g+3-c+3-infty} $c^+_4(r) \leq  \frac{1}{\lambda} + \frac{1}{\beta} + \frac{1}{\lambda\beta} $, for every $r \in V$.
        \item \label{property:g+3-c+3-1} $|c^+_4|_1 \leq \big( \frac{1}{\lambda} + \frac{1}{\beta} + \frac{1}{\lambda\beta} \big)k$.
    \end{enumerate}
\end{lemma}

% \begin{lemma}
% \label{lem:deletededgesatstep2}
% For any graph $G$, any clustering $\mathcal{C}$, any integer $k$ (s.t. $1\le k\le n$),  any subset $U \subset V$ such that $|U| = k$. For any $u_i \in U$, let $\varphi^+_4(u) = \{uv \mid uv \in E(H) \setminus E(\Tilde{G}), \mathcal{C}(u) = \mathcal{C}(v) \}$ be the set of edges that are deleted at step \ref{alg:pre-clusteringsecond} in Algorithm~\ref{alg:pre-clustering} and are not cut in $\mathcal{C}$. Then, 
% \begin{align*}
%     \sum_{u \in U} |\varphi^+_4(u)| \leq \big( \frac{1}{\lambda} + \frac{1}{\beta} + \frac{1}{\lambda\beta} \big)\cdot \cost^k_{\mathcal{C}}
% \end{align*}

% % For a given instance $G$, any clustering $\mathcal{C}$ and any integer $k$ (s.t. $1\le k\le n$), we choose $k$ \textbf{different} vertices $u_1,u_2,u_3,...,u_k \in V$, which form the set $U$. \textbf{Convert all $+$ edges between two light vertices into $-$ edges.} Let the number of edges converted incident to $u_i \in U$ that are not cut in $\mathcal{C}$ to be $\phi(u_i)$. Then, $\sum_{u_i \in U} \phi(u_i)$ is at most $O(1) \cdot \mathbf{OBJ}(G,k,\mathcal{C})$.
% \end{lemma}

\begin{proof}
For any vertex $u \in U$ with $|\varphi^+_4(u)| > 0$, we know that we lose at least $\lambda$ fraction of nodes at step 1. Therefore, let $D(u) \subset \{uv \mid uv \in E_G \setminus E_H \}$ be an arbitrary subset of deleted edges at step 1 with $|D(u)| = \lambda \cdot d(u)$. Then, we have 

\begin{align*}
    g^+_4 &= \sum_{u \in U, uv \in \varphi^+_4} 1 \leq \sum_{u \in U, uv \in \varphi^+_4} \frac{|D(u)|}{\lambda d(u)}
    = \sum_{u \in U, uv \in \varphi^+_4} \frac{1}{\lambda d(u)} \cdot \sum_{uw \in D(u)} 1 \\
    &= \sum_{u \in U, uv \in \varphi^+_4} \frac{1}{\lambda d(u)} \cdot \sum_{uw \in D(u) \cap \varphi^+_2} 1 + \sum_{u \in U, uv \in \varphi^+_4} \frac{1}{\lambda d(u)} \cdot \sum_{uw \in D(u) \setminus \varphi^+_2} 1 \\
    &\leq \sum_{u \in U, uv \in \varphi^+_4} \frac{1}{\lambda d(u)} \cdot \sum_{uw \in D(u) \cap \varphi^+_2} 1 + \sum_{u \in U, uv \in \varphi^+_4} \frac{1}{\lambda d(u)} \cdot \cost_{\calC}(u) \\
    &\leq \sum_{u \in U, uv \in \varphi^+_4} \frac{1}{\lambda d(u)} \cdot \sum_{uw \in D(u) \cap \varphi^+_2} 1 + \frac{1}{\lambda} \sum_{u \in U} \cost_{\calC}(u)
\end{align*}

The second inequality holds because every edge in $D(u) \setminus \varphi^+_2(u)$ will contribute one disagreement to $u$. The third inequality holds because $|\varphi^+_4(u)| \leq d(u)$.

By Lemma \ref{lemma:varphi1}, we know that for any $uw \in \varphi^+_2(u)$, we have 
\begin{align*}
    \frac{1}{\beta} \cdot \big( \frac{\cost_{\mathcal{C}}(u)}{M_{uw}}+ \frac{\cost_{\mathcal{C}}(w)}{M_{uw}}\big) \geq 1
\end{align*}

so, 
\begin{align*}
    g^+_4 &\leq \sum_{u \in U, uv \in \varphi^+_4} \frac{1}{\lambda d(u)}  \sum_{uw \in D(u) \cap \varphi^+_2}\frac{1}{\beta} \big( \frac{\cost_{\mathcal{C}}(u)}{M_{uw}}+ \frac{\cost_{\mathcal{C}}(w)}{M_{uw}}\big) + \frac{1}{\lambda} \sum_{u \in U} \cost_{\calC}(u) \\
    &\leq \frac{1}{\lambda\beta} \sum_{\substack{u \in U, uv \in \varphi^+_4(u) \\ uw \in D(u) \cap \varphi^+_2(u)}} \frac{\cost_{\mathcal{C}}(u)}{d(u) \cdot d(u)}+ \frac{1}{\lambda\beta}\sum_{\substack{u \in U, uv \in \varphi^+_4(u) \\ uw \in D(u) \cap \varphi^+_2(u)}} 
 \frac{\cost_{\mathcal{C}}(w)}{d(u) \cdot M_{uw}} + \frac{1}{\lambda} \sum_{u \in U} \cost_{\calC}(u) \\
     &\leq \frac{1}{\lambda\beta} \sum_{u \in U} \frac{|\varphi^+_4(u)|\cdot |D(u) \cap \varphi^+_2(u)|\cdot \cost_{\mathcal{C}}(u)}{d(u) \cdot d(u)}+ \frac{1}{\lambda\beta} \sum_{\substack{u \in U, uv \in \varphi^+_4(u) \\ uw \in D(u) \cap \varphi^+_2(u)}} 
 \frac{\cost_{\mathcal{C}}(w)}{d(u) \cdot M_{uw}} + \frac{1}{\lambda} \sum_{u \in U} \cost_{\calC}(u) \\
      &\leq \frac{1}{\lambda\beta} \sum_{u \in U} \frac{d(u)\cdot \lambda d(u) \cdot \cost_{\mathcal{C}}(u)}{d(u) \cdot d(u)}+ \frac{1}{\lambda\beta} \sum_{\substack{u \in U, uv \in \varphi^+_4(u) \\ uw \in D(u) \cap \varphi^+_2(u)}} 
 \frac{\cost_{\mathcal{C}}(w)}{d(u) \cdot M_{uw}} + \frac{1}{\lambda} \sum_{u \in U} \cost_{\calC}(u) \\
      &= \frac{1}{\lambda\beta} \cdot \sum_{\substack{u \in U, uv \in \varphi^+_4(u) \\ uw \in D(u) \cap \varphi^+_2(u)}} 
 \frac{\cost_{\mathcal{C}}(w)}{d(u) \cdot M_{uw}}+ \big(\frac{1}{\lambda} + \frac{1}{\beta} \big)\sum_{u \in U} \cost_{\mathcal{C}}(u) \\
 &= \sum_{r \in V} c^+_4(r) \cdot \cost_{\mathcal{C}}(r) 
\end{align*}

The fourth inequality holds because $|\varphi^+_4(u)| \leq d(u)$ and $|D(u) \cap \varphi^+_2(u)| \leq \lambda d(u)$.

   To show \ref{property:g+3-c+3-infty}, we bound the coefficients for $\cost_\calC(u)$ and $\cost_\calC(w)$ respectively.  If $u \in U$, the coefficient for $\cost_\calC(u)$ is $\frac{1}{\lambda} + \frac{1}{\beta}$; if $u \notin U$, the coefficient is $0$.  The coefficient for $\cost_\calC(w)$ is 
   \begin{align*}
\frac{1}{\lambda\beta} \cdot \sum_{\substack{uw \in D(u) \cap \varphi^+_2(u) \\ u \in U, uv \in \varphi^+_4(u) }} \frac{1}{d(u) M_{uw}} 
    &\leq \frac{1}{\lambda\beta} \cdot \sum_{\substack{uw \in D(u) \cap \varphi^+_2(u) \\ u \in U }} \frac{|\varphi^+_4(u)|}{d(u) M_{uw}} 
    &\leq \frac{1}{\lambda\beta} \cdot \sum_{\substack{uw \in D(u) \cap \varphi^+_2(u) \\ u \in U }} \frac{1}{ M_{uw}} \leq \frac{1}{\lambda\beta}
   \end{align*}
   
 Therefore, $c^+_4(r) \leq \frac{1}{\lambda} + \frac{1}{\beta} + \frac{1}{\lambda\beta}$.

    To bound $|c^+_4|_1$, we can replace $\cost_\calC(w)$ and $\cost_\calC(u)$ with 1. Then 
\begin{align*}
   |c^+_4|_1 &= \frac{1}{\lambda\beta} \sum_{\substack{u \in U, uv \in \varphi^+_4(u) \\ uw \in D(u) \cap \varphi^+_2(u)}} 
 \frac{1}{d(u) \cdot M_{uw}} +  \big(\frac{1}{\lambda} + \frac{1}{\beta} \big) \sum_{u \in U} 1\\
    &\leq \frac{1}{\lambda\beta} \sum_{u \in U} \frac{|\varphi^+_4(u)|\cdot |\varphi^+_2(u)|}{d(u) \cdot d(u)} +  \big(\frac{1}{\lambda} + \frac{1}{\beta} \big) k \\
  & \leq \big( \frac{1}{\lambda} + \frac{1}{\beta} + \frac{1}{\lambda\beta} \big)k
\end{align*}
 This proves \ref{property:g+3-c+3-1}. 
\end{proof}

\begin{lemma}
\label{lemma:g+1+2+3}
$g^+_1 + g^+_2 + g^+_4 \leq (\frac{3}{\beta} + \frac{1}{\lambda} + \frac{1}{\beta\lambda} + 1)\cost^k_{\calC}$
\end{lemma}
\begin{proof}
    By Claim \ref{claim:using-coefficients}, Corollary \ref{coro:g+2} and Lemma \ref{lemma:g+3}, we have 
    \begin{align*}
        g^+_1 + g^+_2 + g^+_4 &\leq \cost^k_{\calC} + \sum_{r \in V} c^+_2(r) \cdot \cost_{\calC}(r) + \sum_{r \in V} c^+_4(r) \cdot \cost_{\calC}(r) \\
        &\leq (\frac{3}{\beta} + \frac{1}{\lambda} + \frac{1}{\beta\lambda} + 1)\cost^k_{\calC}
    \end{align*}
\end{proof}

Lemma \ref{lemma:g+1+2+3} gives us a way to bound the cost of deleted edges. Now consider the non-complete graph $G_2$ obtained from $G$ by removing any $+$edge $(u, v)$ (i.e., changing it into a "neutral" edge, and the cost of this edge is $0$) where $u$ and $v$ belong to different connected components of $\tilde{G}$. Note that for any clustering, the cost in $G_2$ is no more than the cost in $G$ since we set some $+$edge costs to 0. Now we are going to bound the cost of $F$ in $G_2$, where $F$ is the clustering output by Algorithm \ref{alg:pre-clustering}. This can give us the relationship between $F$ and $\calC$ in $G_2$. 

Since we will deal with multiple input graphs $G_2$, we include $G_2$ explicitly when necessary. Given a correlation clustering instance $G_2$, an integer $k$, and a clustering $\mathcal{C}$, we denote the disagreement vector in $G_2$ by $\cost_\calC(G_2)$. Consequently, the disagreement for a node $u$ in $G_2$ is specified by $\cost_\calC(G_2, u)$. For a subset $S \subseteq V$, the top-$k$ value on $G_2$ is represented by $\cost^k_\calC(G_2, S) = \max_{T \subseteq S, |T| = k} \sum_{u \in T} \cost_\calC(G_2, u)$. When $k \geq |S|$, then $\cost^k_\calC(G_2, S) = \sum_{u \in S} \cost_\calC(G_2, u)$

The following lemma provides the key insights that $F$ is a good clustering in $G_2$.

\begin{lemma}[Lemma 3.4 of \cite{cohen2021correlation}]
\label{lemma:preclusteringconnecttomostnodesinsameclustering}
    Suppose $5\beta + 2\lambda < 1$. Let $CC \in F$ be a connected component of $\widetilde{G}$ such that $|CC|\ge 2$. Then for each vertex $u\in CC$ we have that 
    \begin{align*}
        d(u,CC)\ge (1-8\beta-\lambda)|CC|.
    \end{align*}
\end{lemma}

Lemma \ref{lemma:preclusteringconnecttomostnodesinsameclustering} tells us that each node $u \in V$ is connected to no less than $\frac{5}{8}$ fractions of the vertices that belong to the same cluster in $F$ when assuming $8\beta + \lambda \leq 3/8$. Now we can bound the cost of $F$ in $G_2$.

\begin{lemma}
\label{lem:nearoptimalindeletedgraph}
Let $G_2$ be a non-complete graph obtained from G by removing any $+$ edge {u, v} (i.e., changing it into a
“neutral” edge) where u and v belong to different connected components of $\Tilde{G}$. Let $\mathcal{C}^*$ be the top-$k$ optimal clustering for graph $G_2$. Assuming $8\beta + \lambda \leq 3/8$. Then, our algorithm outputs solution $F$ such that 
\begin{align*}
    \cost^k_{F}(G_2) \leq 7\cdot \cost^k_{\calC^*}(G_2).
\end{align*}
\end{lemma}
\begin{proof}

Consider any cluster $\tilde{C} \in \mathcal{C}^*$. We claim that the nodes within $\tilde{C}$ originate from the same cluster in $F$. If this were not the case, $\tilde{C}$ could be divided further without increasing the disagreement vector. For a cluster $K$ in $F$, the optimal solution $\mathcal{C}^*$ might partition $K$ into subclusters $K_1, K_2, \ldots, K_l$, arranged such that $|K_1| \geq |K_2| \geq \ldots \geq |K_l|$. Assume that $\cost^k_{F}(G_2)$ involves $t$ nodes from $K$. We aim to demonstrate that:
\begin{align*}
    \cost^t_{F}(G_2, K) &\leq 7 \cdot \cost^t_{\mathcal{C}^*}(G_2, K).
\end{align*}
Applying this relation iteratively across all clusters in $F$ yields the desired result. Consider any node $u$ in $K$. According to Lemma \ref{lemma:preclusteringconnecttomostnodesinsameclustering}, we have:
\begin{align*}
    d(u, K) &\geq \frac{5|K|}{8}.
\end{align*}
Consequently, it follows that:
\begin{align*}
    \cost_{F}(G_2, u) &\leq \frac{3|K|}{8}.
\end{align*}

To show $\cost^t_{F}(G_2, K) \leq 7\cdot \cost^t_{\mathcal{C}^*}(G_2, K)$, we consider two cases: $|K_1| < |K|/2$ or $|K_1| \geq |K|/2$. 

If $|K_1| < |K|/2$, then for any $u \in K_i$, since $d(u, K) \geq \frac{5|K|}{8}$, we have:
\begin{align*}
    d(u, K \setminus K_i) &\geq \frac{5|K|}{8} - |K_i| \\
    &\geq \frac{|K|}{8}.
\end{align*}
This implies that each node $u \in K$ connects to at least $\frac{|K|}{8}$ nodes outside its own cluster in $\mathcal{C}^*$. Consequently, $\cost_{\mathcal{C}^*}(G_2, u) \geq \frac{|K|}{8}$. Therefore, we can conclude:
\begin{align*}
    \cost^t_{F}(G_2, K) &\leq 3 \cdot \cost^t_{\mathcal{C}^*}(G_2, K).
\end{align*}

Now assume $|K_1| \geq \frac{|K|}{2}$. Let $U_1$ be the set of vertices that $\cost^t_{F}(G_2)$ considers in $K_1$, and $U_2$ be the set of vertices it considers in $K \setminus K_1$. 

We analyze the cost increase when combining $K_1, K_2, \ldots, K_l$ into $K$. First, we focus on nodes from $U_2$. Since $|K_1| \ge \frac{|K|}{2}$, the sizes of $K_2, K_3, \ldots, K_l$ are each no more than $\frac{|K|}{2}$. For any node $u$ in $K \setminus K_1$, we have:
\begin{align*}
    \cost_{\mathcal{C}^*}(G_2, u) &\geq d(u, K) - d(u,K\setminus K_1) \geq d(u, K) - \frac{|K|}{2} \geq \frac{|K|}{8}.
\end{align*}
For any node in $U_2$, we have:
\begin{align*}
    \cost_{F}(G_2, u) &\leq \frac{3|K|}{8}.
\end{align*}
Therefore, the maximum increase in cost for $U_2$ when we combine $K_1, K_2, \ldots, K_l$ into $K$ is at most:
\begin{align*}
    \cost^t_{F}(G_2, U_2) - \cost^t_{\mathcal{C}^*}(G_2, U_2) &\leq 2 \cdot \cost^t_{\mathcal{C}^*}(G_2, U_2) \\
    &\leq 2 \cdot \cost^t_{\mathcal{C}^*}(G_2, K).
\end{align*}

The challenging case concerns the increase in cost due to $U_1$. We consider the $-$edges between $U_1$ and $K \setminus K_1$, as these edges will lead to an increased cost for $U_1$. When $|U_1| \leq \frac{3|K|}{8}$, the number of $-$edges between $U_1$ and $K \setminus K_1$ is at most:
\begin{align*}
    %|K \setminus K_1| \leq 
    (|K| - |K_1|) \cdot |U_1|.
\end{align*}
When $|U_1| > \frac{3|K|}{8}$, for each node $u \in K \setminus K_1$, we have:
\begin{align*}
    d(u, U_1) &\geq \frac{5|K|}{8} - (|K| - |U_1|) \\
    &\geq |U_1| - \frac{3|K|}{8}.
\end{align*}
Thus, the number of $+$edges between $U_1$ and $K \setminus K_1$ is at least:
\begin{align*}
    \sum_{u\in U_1}d(u, K \setminus K_1) &= \sum_{v \in K \setminus K_1}d(v, U_1) \\
    &\geq (|K| - |K_1|) (|U_1| - \frac{3|K|}{8}),
\end{align*}
and the number of $-$edges between $U_1$ and $K \setminus K_1$ is at most:
\begin{align*}
    |U_1| \cdot |K \setminus K_1| - \sum_{u\in U_1}d(u, K \setminus K_1) &\leq (|K| - |K_1|) \cdot |U_1| - (|K| - |K_1|) (|U_1| - \frac{3|K|}{8}) \\
    &\leq (|K| - |K_1|) \cdot \frac{3|K|}{8}.
\end{align*}
Therefore, the increase in disagreement for $U_1$ is at most:
\begin{align*}
    \cost^t_{F}(G_2, U_1) - \cost^t_{\mathcal{C}^*}(G_2, U_1) &\leq (|K| - |K_1|) \cdot \min(|U_1|, \frac{3|K|}{8}).
\end{align*}
Now we connect this increase to the cost of $K$ in $\calC^*$. When $|U_1| \geq |K \setminus K_1|$, we have 
\begin{align*}
    \cost^t_{\mathcal{C}^*}(G_2, K) &\geq \sum_{u \in K \setminus K_1} \cost_{\mathcal{C}^*}(G_2, u) \\
    &\geq (|K| - |K_1|) \frac{|K|}{8}.
\end{align*}
The second inequality holds since for each node $u \in K\setminus K_1$, we have $\cost_{\mathcal{C}^*}(G_2, u) \geq |K| / 8$. This implies:
\begin{align*}
    \cost^t_{F}(G_2, U_1) - \cost^t_{\mathcal{C}^*}(G_2, U_1) &\leq 3\cost^t_{\mathcal{C}^*}(G_2, K).
\end{align*}
When $|U_1| < |K \setminus K_1|$, if we consider $|U_1|$ nodes from $K \setminus K_1$ for $\mathcal{C}^*$, we find:
\begin{align*}
    \cost^t_{\mathcal{C}^*}(G_2, K) &\geq |U_1| \cdot \frac{|K|}{8}.
\end{align*}
On the other side, we have
\begin{align*}
    \cost^t_{F}(G_2, U_1) - \cost^t_{\mathcal{C}^*}(G_2, U_1) &\leq (|K| - |K_1|) \cdot \min(|U_1|, \frac{3|K|}{8}) \\
    &\leq \frac{|K|}{2} \cdot |U_1|.
\end{align*}
The second inequality holds since $|K_1| \geq |K| / 2$. Combining these cases, the increase for $U_1$ is at most $4\cost^t_{\mathcal{C}^*}(G_2, K)$ and:
\begin{align*}
    \cost^t_{F}(G_2, K) - \cost^t_{\mathcal{C}^*}(G_2, K) &\leq 
    \cost^t_{F}(G_2, U_1) - \cost^t_{\mathcal{C}^*}(G_2, U_1) + \cost^t_{F}(G_2, U_2) - \cost^t_{\mathcal{C}^*}(G_2, U_2) \\
    &\leq 4\cost^t_{\mathcal{C}^*}(G_2, K) + 2\cost^t_{\mathcal{C}^*}(G_2, K) \\
    &\leq 6\cost^t_{\mathcal{C}^*}(G_2, K)
\end{align*}
providing the target statement.
\end{proof}

Now we can show our main lemma regarding the approximate ratio.

\begin{proof}[Proof of Lemma \ref{lemma:preclusteringratio}]
 Recall that $\calC$ and $\calC^*$ represent the top-$k$ optimal clustering for graph $G$ and graph $G_2$, respectively. Let $\opt^k$ denote the cost of the optimum solution under the top-$k$ norm for $G$. When we transfer the graph $G_2$ to $G$, we will incur an increase in the cost of $F$. Lemma \ref{lemma:g+1+2+3} provides the upper bounds on the increase in the cost of $F$, given by:
\begin{align*}
    \cost^k_{F}(G) &\leq \cost^k_{F}(G_2) + g^+_1 + g^+_2 + g^+_4 \\
    &\leq \left(\frac{3}{\beta} + \frac{1}{\lambda} + \frac{1}{\beta\lambda} + 1\right)\opt^k + \cost^k_{F}(G_2).
\end{align*}
On the other hand, based on Lemma \ref{lem:nearoptimalindeletedgraph}, we know that in $\tilde{G}$, $F$ is a $7$-approximate solution in $G_2$. Thus, we have:
\begin{align*}
    \cost^k_{F}(G_2) &\leq 7 \cost^k_{\mathcal{C}^*}(G_2) 
    \leq 7 \cost^k_{\calC}(G_2) 
    \leq 7\opt^k,
\end{align*}
which establishes the desired inequality.
\end{proof}

\bibliographystyle{alpha}
\bibliography{local}

\appendix

\section{Reduction from All Monotone Symmetric Norms to Top-$k$ Norms}
\label{sec:all-norm}
% \begin{definition}[Top-$k$ norm]
%     Given a norm $f$, for any vector $x\in \mathbb{R}^n$, $f(x)$ is equal to the sum of the $k$ largest coordinates of $x$, then $f$ is called the top-$k$ norm. For simplicity, let $\topk(x)$ denote the top-$k$ norm of $x$.
% \end{definition} 

\begin{definition}[Ordered Norms]
    For any vector $x\in \mathbb{R}_{\geq 0}^n$, let $x^\downarrow$ denote the vector $x$ with its coordinates sorted in non-increasing order. Given weight vector $w \in \R_{\geq 0}^n$ with $w_1 \geq w_2 \geq \cdots \geq w_n$, the $w$-ordered norm of $x$ is defined as $\order(w;x)=\sum_{i=1}^n w_i x_i^\downarrow$.
\end{definition}

\begin{lemma}[Lemma 5.2 of \cite{chakrabarty2019approximation}]\label{lem:Ordered2all}
    For any monotone and symmetric norm $f:\mathbb{R}^n\rightarrow \mathbb{R}_+$, define the set $\mathbb{B}_+(f):=\{x\in \mathbb{R}_+^n:f(x)\le 1\}$, and $W=\{w\in \mathbb{R}_+^n : w_1\ge w_2 \ge \cdots \ge w_n, w\ is\ a\ subgradient\ of\\ f\ at\ some\ x\in \mathbb{B}_+(f)\}$. Then we have $f(x)=\max_{w\in W} \order(w;x)$ for every $x\in \mathbb{R}_{\geq 0}^n$.
\end{lemma}

\lemmatopktolpnorm*
\begin{proof}[Proof of Lemma~\ref{lem:Topk2all}]
    For any $w=(w_1,w_2,\dots,w_n)$ such that $w_1\ge w_2 \ge \cdots \ge w_n\ge 0$, if we set $w'=(w'_1,w'_2,\dots,w'_n)$ as 
    \begin{align*}
    w_i'= \left\{
        \begin{array}{rcl}
            w_i-w_{i+1} & & {i\in [1,n-1]}\\
            w_n & & {i=n}
        \end{array}
        \right.
    \end{align*}

    Let $\topk(x)$ denote the top-$k$ norm of $x$. Then we have $\order(w;x)=\sum_{k=1}^n w_k'\cdot \topk(x)$.

    Let $\mathbb{B}_+(f):=\{x\in \mathbb{R}_+^n:f(x)\le 1\}$ and $W=\{w\in \mathbb{R}_+^n : w_1\ge w_2 \ge \cdots \ge w_n, w\ is\ a\ subgradient\ of\\ f\ at\ some\ x\in \mathbb{B}_+(f)\}$. We construct a new set $W'=\{w'|w\in W : w'=(w'_1 = w_1-w_2, w'_2 = w_2-w_3,\dots,w'_{n- 1} = w_{n-1}-w_n, w'_n = w_n) \}$. By Lemma~\ref{lem:Ordered2all}, we have $f(x)=\max_{w\in W} \order(w;x)=\max_{w'\in W'}\sum_{k=1}^n w_k'\cdot \topk(x)$.

    Let $y$ be the disagreement vector for the given clustering $\calC_{ALG}$. For any symmetric monotone norm $f: \R_{\geq 0}^n \to \R_{\geq 0}$, define $y^*_f$ to be the disagreement vector for the optimal clustering under the norm $f$. By the assumption that $\calC_{ALG}$ is a simultaneous $\rho$-approximation for every top-$k$ norm, we have we have $\topk(y)\le \rho \cdot \topk(y^*_{\text{top-k}})$ for every $k \in [n]$, where $y^*_{\text{top-k}}$ is the disagreement vector for the optimal clustering under top-k objective. Now we bound $f(y)$ in terms of $f(y^*_f)$ for any monotone symmetric norm $f$:
    \begin{align*}
        f(y)&=\max_{w'\in W'}\sum_{k=1}^n w_k'\cdot \topk(y)
        \le \max_{w'\in W'}\sum_{k=1}^n w_k'\cdot \rho \cdot \topk(y^*_{\text{top-k}})
        = \rho \cdot \max_{w'\in W'}\sum_{k=1}^n w_k'\cdot \topk(y^*_{\text{top-k}})\\
        &\le \rho \cdot \max_{w'\in W'}\sum_{k=1}^n w_k'\cdot \topk(y^*_f)
        = \rho \cdot f(y^*_f). \qedhere
    \end{align*}
\end{proof}
% {\color{gray}
% By lemma~\ref{lem:Topk2all}, we conclude that the single clustering returned by algorithm~\ref{alg:pre-clustering} is simultaneously a $(\frac{3}{\beta} + \frac{1}{\lambda} + \frac{1}{\beta\lambda}  + 4)$-approximation for all monotone and symmetric norm objectives, and the single clustering returned by algorithm~\ref{alg:pre-clusteringlp} with KMZ rounding algorithm is simultaneously a 60-approximation for all monotone and symmetric norm objectives.
% }

\section{Proof for Theorem \ref{theorem:constructHbysampling}}
\label{sec:constructH}
We repeat Theorem \ref{theorem:constructHbysampling} for convenience. 
\thmcnostructH*

\begin{proof}

The algorithm \ref{alg:constrcuctHbysampling} works as follows: For any edge $uv \in E$, assume $d(u) \geq d(v)$ without loss of generality. Let $j_u$ be the $j$ such that $2^{j-1} < d(u) \leq 2^j$. We then sample nodes from $N(u) \Delta N(v)$ with probability $\min(\frac{\tau}{\beta 2^{j_u}}, 1)$ (Line 5). Let $X_{uv}$ be the number of sampled nodes from $N(u) \Delta N(v)$. If $X_{uv}$ is at least $(1 + \frac{\epsilon}{2}) \cdot \tau \frac{d(u)}{2^j}$, we set $I(uv)$ to $0$ since there are too many sampled nodes. Otherwise, $I(uv)$ is set to $1$ (Lines 6-11).

Note that $d(u) \leq 2^j$, so in expectation, we will only sample $O(\tau / \beta)$ nodes from $N(u)$ and $N(v)$, which bounds the number of sampled nodes and total number of memory.

\begin{algorithm}[ht!]
\caption{Construct graph H.\\
\textbf{Input}: Graph $G = (V, E)$\\
\textbf{Output}: For any edge $uv \in E$, set $I(uv) = 1$ if $|N(u) \Delta N(v)| \leq \beta \max(d(u), d(v))$; set $I(uv) = 0$ if $|N(u) \Delta N(v)| \geq (1+\epsilon)\beta \max(d(u), d(v))$. %A partition of $V$, $\mathcal{S}$.
}
\label{alg:constrcuctHbysampling}
\begin{algorithmic}[1]
\Function {\textsc{Agreement}}{$G = (V, E)$}
    \State $\tau \leftarrow \frac{400 \log n}{\epsilon^2}$
    \For{$j \in [0, \log n]$}
    \label{alg:AllNormCCBySamplingsetlpvaluestart-1}
    \State $ S(j) \leftarrow \emptyset$
    \For{$v \in V$} add $v$ to $S(j)$ with probability $\min(\tau / (\beta2^j), 1)$
    \EndFor
    \For{$uv \in E$ } 
    \If{$2^{j-1} < M_{uv} \leq 2^{j}$} \Comment{$M_{uv} = \max(d(u), d(v))$}
    \State $X_{uv} \leftarrow |(N(u) \cap S(j)) \Delta (N(v) \cap S(j))|$
    \If{$X_{uv} > (1 + \frac{\epsilon}{2}) \cdot \tau \frac{M_{uv}}{2^j}$}
        \State $I(uv) \leftarrow 0$
    \Else
    \State $I(uv) \leftarrow 1$
    \EndIf
    \EndIf
    \EndFor
    \EndFor
    \EndFunction
\end{algorithmic}
\end{algorithm}

To show that Algorithm \ref{alg:constrcuctHbysampling} correctly outputs, assume $\tau / \beta$ is a power of 2 for simplicity. When $d(u) \leq \tau / \beta$, we have $2^{j_u} \leq \tau / \beta$ and $\frac{\tau}{\beta 2^{j_u}} \geq 1$. Hence, we will set $S(j)$ as $V$ and $X_{uv} = |N(u) \Delta N(v)|$ is the exact estimate. Now assume $d(u) > \tau / \beta$. We have $\tau / (\beta 2^{j_u}) < 1$ and by linearity of expectation, we have
\begin{align*}
    \E[X_{uv}] = \frac{\tau}{\beta 2^{j_u}} \cdot |N(u) \Delta N(v)|.
\end{align*}

If $|N(u) \Delta N(v)| > (1 + \epsilon)\beta \max(d(u), d(v))$, then 
\begin{align*}
    \E[X_{uv}] > \frac{\tau}{\beta 2^{j_u}} \cdot (1 + \epsilon)\beta d(u) = (1 + \epsilon) \tau \cdot \frac{d(u)}{2^{j_u}}.
\end{align*}
Note that $d(u) \geq 2^{j_u} / 2$. Hence, by the Chernoff bound \ref{thm:chernoffleq}, we have 
\begin{align*}
    \Pr[X_{uv} \leq (1 + \frac{\epsilon}{2}) \cdot \tau \frac{d(u)}{2^{j_u}}] \leq \exp\left(-\left(\frac{\epsilon}{2}\right)^2 \cdot \frac{\tau}{8}\right) \leq n^{-10}.
\end{align*}

On the other side, when $|N(u) \Delta N(v)| \leq \beta \max(d(u), d(v))$, we have 
\begin{align*}
    \E[X_{uv}] \leq \frac{\tau}{\beta 2^{j_u}} \cdot \beta d(u) = \tau \cdot \frac{d(u)}{2^{j_u}}.
\end{align*}
Hence, by the Chernoff bound \ref{thm:chernoffgeq}, we have 
\begin{align*}
    \Pr[X_{uv} > (1 + \frac{\epsilon}{2}) \cdot \tau \frac{d(u)}{2^{j_u}}] \leq \exp\left(-\left(\frac{\epsilon}{2}\right)^2 \cdot \frac{\tau}{6}\right) \leq n^{-10}.
\end{align*}
Thus, by the union bound, we have for all $uv \in E$, the algorithm outputs “Yes” if $|N(u) \Delta N(v)| \leq \beta \max(d(u), d(v))$, and outputs “No” if $|N(u) \Delta N(v)| > (1 + \epsilon) \beta \max(d(u), d(v))$ with probability at least $1 - \frac{1}{n^6}$.

The last thing we want to show is the total memory and total work. All computations are naturally parallelized. We only need to bound the number of sampled nodes for each $u$. Note that when we consider node $u$, we must have $2^{j_u} \geq d(u)$. Let $Y_u$ be the number of sampled nodes when we consider $u$, then 
\begin{align*}
    \E[Y_u] = \frac{\tau}{\beta 2^{j_u}} \cdot d(u) \leq \frac{\tau}{\beta}.
\end{align*}
Thus, by the Chernoff bound, we have 
\begin{align*}
    \Pr[Y_u \geq \frac{2\tau}{\beta}] \leq \exp\left(-\frac{\tau}{3\beta}\right) \leq n^{-10}.
\end{align*}
Hence, the total memory and total work required is $\tilde{O}(m \tau) = \tilde{O}(m / \epsilon^2)$.

\end{proof}

\section{Proof for Lemma \ref{lemma:samplingfinalvalue}}
\label{sec:proofofsamplingfinalvalue}
We repeat Theorem \ref{lemma:samplingfinalvalue} for convenience.
\lemmasmpaling*

\begin{proof}

Assume that $d(u) \geq d(v)$ without loss of generality and $\tau$ is a power of 2 for simplicity. Let $j_u$ be the $j$ such that $2^{j-1} < d(u) \leq 2^j$. We then sample nodes from $N_H(u) \cap N_H(v)$ with probability $\min(\frac{\tau}{\beta 2^{j_u}}, 1)$. 

% Let $X_{uv}$ be the number of sampled nodes from $N(u) \Delta N(v)$. 
% If $X_{uv}$ is at least $(1 - \frac{\epsilon}{2}) \cdot \tau \frac{d(u)}{2^{j_u}}$, we set $I(uv)$ to 0 since there are too many sampled nodes. Otherwise, $I(uv)$ is set to 1 (Lines 6-11).
If $|N_H(u)| \leq \tau$, then when we consider $uv$, we have $2^{j_u} \leq \tau$ and we sample each node with probability $\tau / 2^{j_u} \geq 1$. Therefore, we will compute $|N_H(u) \cap N_H(v)|$ exactly and $\tilde{x}_{uv} = x_{uv}$. Otherwise, we have $d_H(u) > \tau$. Define the random variables:
\begin{align*}
    W_{uv} &= \frac{2^{j_u}}{\tau}\sum_{w\in N_H(u) \cap N_H(v)}\mathbb{1}(w\in S(j_u))
\end{align*}
so we have $\E[W_{uv}] = |N_H(u) \cap N_H(v)|$ and $\tilde{x}_{uv} = 1 - \frac{W_{uv}}{d(u)}$. We will divide $|N_H(u) \cap N_H(v)|$ into 3 cases, and in each case, we will show our argument holds with probability at least $1 - 1 / n^6$.

The first case is $|N_H(u) \cap N_H(v)| \geq (1 - \frac{\epsilon}{2}) \cdot d(u)$, we have $x_{uv} \leq \frac{\epsilon}{2}$ and 
\begin{align*}
    \E[W_{uv}] = |N_H(u) \cap N_H(v)| \geq (1 - \frac{\epsilon}{2}) \cdot d(u).
\end{align*}
By Chernoff bound \ref{thm:chernoffgeq}, we have 
\begin{align*}
    \Pr[W_{uv} \leq (1 - \epsilon) \cdot d(u)] &= \Pr[W_{uv} \leq \frac{1 - \epsilon}{1 - \epsilon / 2} \cdot (1 - \frac{\epsilon}{2}) \cdot d(u)] \\ 
    &\leq \Pr[W_{uv} \leq (1 - \frac{\epsilon}{2}) \cdot (1 - \frac{\epsilon}{2}) \cdot d(u)] \\
    &\leq \exp \left(- (\frac{\epsilon}{2})^2 \cdot \frac{d(u)}{4} \right) \\
    &\leq \exp \left(- \frac{\tau}{16} \right) \leq \exp(-10\log n) \leq n^{-10}.
\end{align*}

Thus, we also have $\tilde{x}_{uv} = 1 - \frac{W_{uv}}{d(u)} \leq \epsilon$, with probability at least $1 - n^{-8}$ for all $uv$. If $uv \in E$, at line \ref{alg:AllNormCCBySamplingRoundingstart}, we will set $\tilde{x}_{uv}$ to 0 and $\tilde{x}_{uv} \leq x_{uv}$. If $uv \in {V \choose 2} \setminus E$, we have $1 - \tilde{x}_{uv} \leq 1 \leq (1 + \epsilon)(1 - x_{uv})$. So Lemma \ref{lemma:samplingfinalvaluelpvalue} holds when $|N_H(u) \cap N_H(v)| \geq (1 - \frac{\epsilon}{2}) \cdot d(u)$.

The second case is $|N_H(u) \cap N_H(v)| \leq \frac{\epsilon}{2} \cdot d(u)$. Again, we have 
\begin{align*}
    \E[W_{uv}] = |N_H(u) \cap N_H(v)| \leq \frac{\epsilon}{2} \cdot d(u).
\end{align*}
By Chernoff bound \ref{thm:chernoffleq}, we have
\begin{align*}
    \Pr[W_{uv} \geq \epsilon \cdot d(u)] &= \Pr[W_{uv} \geq (1 + 1) \cdot \frac{\epsilon}{2} \cdot d(u)] \\ 
    &\leq \exp \left(- \frac{\epsilon d(u)}{6} \right) \\
    &\leq \exp \left(- \frac{\tau}{6} \right) \leq \exp(-10\log n) \leq n^{-10}.
\end{align*}
Thus, we also have $\tilde{x}_{uv} \geq 1 - \epsilon$, with probability at least $1 - n^{-8}$ for all $uv$. Note that $x_{uv} \geq 1 - \frac{\epsilon}{2}$. If $uv \in E$, we have $\tilde{x}_{uv} \leq 1 \leq (1 + \epsilon) x_{uv}$. If $uv \in {V \choose 2} \setminus E$, at line \ref{alg:AllNormCCBySamplingRoundingmid}, we will set $\tilde{x}_{uv}$ to 1 and $1 - \tilde{x}_{uv} \leq (1 + \epsilon)(1 - x_{uv})$.

The remaining case is $|N_H(u) \cap N_H(v)| \in \left( \frac{\epsilon}{2} \cdot d(u), (1 - \frac{\epsilon}{2}) \cdot d(u) \right)$. In this case, we have
\begin{align*}
    \Pr[\tilde{x}_{uv} > (1 + \epsilon)x_{uv}] &= \Pr[1 - \frac{W_{uv}}{d(u)} > (1 + \epsilon)(1 - \frac{|N_H(u) \cap N_H(v)|}{d(u)})] \\
    &= \Pr[\frac{W_{uv}}{d(u)} + \epsilon \leq (1 + \epsilon) \frac{|N_H(u) \cap N_H(v)|}{d(u)}] \\
\end{align*}

Note that $|N_H(u) \cap N_H(v)| < (1 - \frac{\epsilon}{2}) \cdot d(u)$, so $\frac{|N_H(u) \cap N_H(v)|}{(1 - \frac{\epsilon}{2}) \cdot d(u)} < 1$, 
and
\begin{align*}
    \Pr[\tilde{x}_{uv} > (1 + \epsilon)x_{uv}] &\leq \Pr[\frac{W_{uv}}{d(u)} + \epsilon \cdot \frac{|N_H(u) \cap N_H(v)|}{(1 - \frac{\epsilon}{2}) \cdot d(u)} \leq (1 + \epsilon) \frac{|N_H(u) \cap N_H(v)|}{d(u)}] \\
    &\leq \Pr[W_{uv} \leq (1 + \epsilon - \frac{\epsilon}{1 - \epsilon/2}) |N_H(u) \cap N_H(v)|] \\
    &\leq \Pr[W_{uv} \leq (1 - \frac{\epsilon^2}{2 - \epsilon}) |N_H(u) \cap N_H(v)|] 
\end{align*}
Since $\E[W_{uv}] = |N_H(u) \cap N_H(v)|$, by chernoff bound \ref{thm:chernoffleq}, we have 
\begin{align*}
    \Pr[\tilde{x}_{uv} > (1 + \epsilon)x_{uv}] &\leq \exp\left( -\left(\frac{\epsilon^2}{2}\right)^2 \frac{|N_H(u) \cap N_H(v)|}{3}\right) \\
    &\leq \exp\left( -\frac{\epsilon^4}{4} \frac{\epsilon d(u)}{6}\right) \leq n^{-10}.
\end{align*}

Similarly, we have
\begin{align*}
    \Pr[1 - \tilde{x}_{uv} > (1 + \epsilon)(1 - x_{uv})] &= \Pr[\frac{W_{uv}}{d(u)} > (1 + \epsilon) \frac{|N_H(u) \cap N_H(v)|}{d(u)}] \\
    &\leq \Pr[W_{uv} > (1 + \epsilon) |N_H(u) \cap N_H(v)|] \\
    &\leq \exp\left( -\epsilon^2 \frac{|N_H(u) \cap N_H(v)|}{2}\right) \\
    &\leq \exp\left( -\epsilon^2 \frac{\epsilon d(u)}{2}\right) \leq n^{-10}.
\end{align*}
Thus, when $|N_H(u) \cap N_H(v)| \in \left( \frac{\epsilon}{2} \cdot d(u), (1 - \frac{\epsilon}{2}) \cdot d(u) \right)$, Lemma \ref{lemma:samplingfinalvaluelpvalue} holds for any edge from $E \cup K$. There might be some $-$edges in $\binom{V}{2} \setminus (E \cup K)$ with $x > 1 - \epsilon$, and for these edges, we set their $\widetilde{x}$ value as 1, resulting in zero cost. Combining all together, we know that Lemma \ref{lemma:samplingfinalvaluelpvalue} holds for all $uv$, with probability at least $1 - n^{-6}$.

For Lemma \ref{lemma:samplingfinalvaluetriangle}, consider edges from $\binom{V}{2} \setminus (E \cup K)$. We know that $x_{uv} \geq 1 - \epsilon$ and we set $\tilde{x}_{uv}$ as 1. For edges from $E \cup K$, we have three cases. For the first and second cases, we have $x_{uv} - \epsilon \leq \tilde{x}_{uv} \leq x_{uv} + \epsilon$. For the third case, we already have $\tilde{x}_{uv} < (1 + \epsilon)x_{uv} < x_{uv} + \epsilon$ with high probability. For the other direction,
\begin{align*}
    \Pr[\tilde{x}_{uv} < x_{uv} - \epsilon] &\leq \Pr[\tilde{x}_{uv} < (1 + \epsilon)x_{uv} - \epsilon] \\
    &= \Pr[1 - \tilde{x}_{uv} > 1 + \epsilon - (1 + \epsilon)x_{uv}] \\
    &\leq n^{-10}.
\end{align*}
In either case, we have $x_{uv} - \epsilon < \tilde{x}_{uv} < x_{uv} + \epsilon$ holding with probability at least $1 - n^{-6}$ for all $uv$. Hence, the approximate triangle inequality holds with high probability.

\end{proof}

% \section{Primal Dual}
% \input{PrimalDual}

% \input{Notes-Shili-May82024}

\end{CJK*}\end{document}